\documentclass{fundam}

\publyear{25}
\papernumber{4}
\volume{195}
\issue{1-4}
\theDOI{10.46298/fi.12451}

\versionForARXIV

\usepackage[all]{xy}
\usepackage{tikz}
\usetikzlibrary{matrix,arrows}
\usepackage[mathscr]{eucal}
\usepackage{amssymb}
\usepackage{stmaryrd}
\usepackage{enumerate,graphicx}
\usepackage{url} 
\usepackage{multicol}
\usepackage{marginnote}
\numberwithin{equation}{section}

\usepackage{pgf,tikz}

\usetikzlibrary{calc,positioning,shapes,arrows.meta}

\tikzset{%
	shaded/.style={draw, shape=circle, fill=black!35, inner sep=1.4pt},
	unshaded/.style={draw, shape=circle, fill=white, inner sep=1.4pt},
	quasi/.style={draw, shape=rectangle, rounded corners=3pt, fill=white, inner sep=2.5pt, minimum height=14.5pt},
	blob/.style={draw, shape=rectangle, rounded corners=12pt, thin, densely dotted},
	arrow/.style={->, thin, >=latex, shorten >=2.5pt, shorten <=2.5pt},
	order/.style={thin},
	map/.style={->, densely dashed, shorten >=5pt, shorten <=5pt, >=stealth', looseness=1.1},
	curvy/.style={thin, looseness=1.2, bend angle=70},
	fatcurvy/.style={thin, looseness=1.7, bend angle=75},
	label/.style={shape=rectangle, inner sep=6pt},
	auto}
\font\bmi=cmmi8 scaled 1440
\newcommand{\powerset}{\raise.6ex\hbox{\bmi\char'175 }}

\newcommand{\bigand}{\mathop{\bigwedge\kern -8.5truept \bigwedge}}
\newcommand{\Bigand}{\mathop{\bigwedge\kern -10truept \bigwedge}}
\newcommand{\littleand}{\mathbin{\wedge\kern -8truept \wedge}}
\newcommand{\bigor}{\mathop{\bigvee\kern -8.5truept \bigvee}}
\newcommand{\Bigor}{\mathop{\bigvee\kern -10truept \bigvee}}
\newcommand{\littleor}{\mathbin{\vee\kern -8truept \vee}}
\renewcommand{\le}{\leqslant}
\renewcommand{\leq}{\leqslant}

\renewcommand{\preceq}{\preccurlyeq}

\begin{document}
	
	%%%%%%%%%%%%%%%%%%%%%%%%%%%%%%%%%%%%%%%%%%%%%%%%%%%%%%%%%%%%%%%%%%%%%%
	%% FRONT MATTER
	%%%%%%%%%%%%%%%%%%%%%%%%%%%%%%%%%%%%%%%%%%%%%%%%%%%%%%%%%%%%%%%%%%%%%%

\runninghead{A. Craig and C. Robinson}{Representing Sugihara monoids}
 
\title{Representing Sugihara monoids via weakening relations}
	\address{Department of Mathematics and Applied Mathematics, University of Johannesburg, PO Box 524, Auckland Park 2006, South Africa}
	
	\author{Andrew Craig \\
 Department of Mathematics and Applied Mathematics\\
 University of Johannesburg, South Africa \\and \\
 National Institute for Theoretical and Computational Sciences (NITheCS) \\ Johannesburg, South Africa \\ acraig@uj.ac.za\thanks{Funding from the National Research Foundation grant 127266 is gratefully acknowledged.}
		\and Claudette Robinson \\
		Department of Mathematics and Applied Mathematics\\
		University of Johannesburg, South Africa\\
		claudetter@uj.ac.za
	}
	
	\maketitle
	
\begin{abstract}
We show that all Sugihara monoids can be represented as algebras of binary relations, with the monoid operation given by relational composition. Moreover, the binary relations are weakening relations. 
The first step is to obtain an explicit relational representation of all finite odd Sugihara chains. Our construction mimics that of Maddux~\cite{Mad10}, where a relational representation of the finite even Sugihara chains is given. 
We define the class of representable Sugihara monoids as those which can be represented as so-called direct reducts of distributive involutive FL-algebras of binary relations. We then show that the class of  representable distributive involutive FL-algebras is closed under ultraproducts. 
This fact is used to demonstrate that the two infinite Sugihara monoids that generate the quasivariety are also representable. From this it follows that all Sugihara monoids are representable. 
\end{abstract}
	
%%%%%%%%%%%%%%%%%%%%%%%%%%%%%%%%%%%%%%%%%%%%%%%%%%%%%%%%%%%%%%%%%%%%%%
	%% MAIN MATTER
%%%%%%%%%%%%%%%%%%%%%%%%%%%%%%%%%%%%%%%%%%%%%%%%%%%%%%%%%%%%%%%%%%%%%%
	
\section{Introduction}\label{sec:intro}
	
Sugihara monoids and Sugihara algebras have found many applications in the study of relevance logics. The variety of Sugihara algebras is generated by 
$\mathbf{Z}=
\langle \mathbb{Z},\wedge,\vee,\cdot,\to,
%1, 
{\sim}\rangle $ 
and provides complete algebraic semantics for the logic $\mathbf{R}$-mingle (often denoted $\mathbf{RM}$)~\cite{AB75, D70}. 
Sugihara monoids can be described as   commutative distributive idempotent residuated lattices with an order-reversing involution. That is, a truth constant, which is the identity element of the monoid, is included in the signature.  
They are known to give algebraic semantics for $\mathbf{RM^t}$, the logic $\mathbf{R}$-mingle with added Ackermann constants~\cite{AB75}. 
	
Odd Sugihara monoids are those for which the monoid identity 
(truth constant)
is fixed by the involution. They were studied by Galatos and Raftery~\cite{GR12} where they showed a categorical equivalence with relative Stone algebras. An application provided results about Beth definability and interpolation for the uninorm-based logic $\mathbf{IUML}$. 
Dual topological representations of Sugihara monoids were studied extensively by Fussner and Galatos~\cite{FG19} (see also Fussner's PhD thesis~\cite{Fussner-PhD}). Dual spaces were constructed in two ways---making use of of both Priestley spaces and Esakia spaces. 

Our goal in this paper is to show that it is possible to find a representation of any Sugihara monoid, with the monoid operation given by relational composition. We do this by first giving an
explicit relational representation for each finite odd Sugihara chain. Given a poset $\langle X,\leqslant\rangle $, a binary relation $R$ on $X$ is said to be a
\emph{weakening relation} if ${\leqslant} \circ R \circ {\leqslant} =R$. 
The elements of our algebra will be weakening relations on a poset which 
comes from a twisted product of a poset with its dual. 
Our construction essentially combines the method of Maddux~\cite{Mad10} with recent work by the current authors~\cite{CR23}. There we were able to represent so-called distributive quasi relation algebras (DqRAs) where the elements of the algebra are binary relations and the monoid operation is given by relational composition. Quasi relation algebras (qRAs) were introduced by Galatos and Jipsen~\cite{GJ13} as a generalisation of relation algebras. 
A major 
feature of these algebras 
is that the variety of qRAs has a decidable equational theory (whereas the variety of relation algebras does not). 

In 2010, Maddux~\cite[Theorem 6.2]{Mad10} gave a relational  representation of the finite even Sugihara chains $\mathbf{S}_{2n+2}$ ($n \in \mathbb{Z}^+$). We outline that representation here in very general terms, to help us describe the modification we 
make in order to 
represent the finite odd Sugihara chains. 
For $n \in \mathbb{Z}^+$, consider the set $\mathbb{Q}^n$ with the lexicographic ordering, denoted $\leqslant_n$. Now $\langle \mathbb{Q}^n, \leqslant_n\rangle $ is a chain. With ${\mathbf{X}=\langle X,\leqslant\rangle} =
\langle\mathbb{Q}^n, \leqslant_n\rangle$, consider the usual product poset $\mathbf{X}^\partial \times \mathbf{X}$ where $\mathbf{X}^\partial=\langle X, \geqslant\rangle$. It is then possible to embed $\mathbf{S}_{2n+2}$ in the set of upsets of
$\mathbf{X}^\partial \times \mathbf{X}$ in such a way that the elements of the algebra are binary relations and the monoid operation is given by relational composition. This construction, which is just a re-casting of~\cite{Mad10}, is described in more detail in Example~\ref{ex:Maddux2010}.

For the odd Sugihara chain $\mathbf{S}_{2n+3}$ (where $n \in \mathbb{Z}^+$), we will take two disjoint copies of $\langle\mathbb{Q}^n,\leqslant_n\rangle$, with a map $\alpha$ swapping the occurrence of each element from one copy to the other. This map $\alpha$ is an order isomorphism of the poset consisting of two chains. 
The algebra $\mathbf{S}_{2n+3}$ is then embedded in the upsets of 
$\mathbf{X}^{\partial} \times \mathbf{X}$, where this time $\mathbf{X}$ consists of two disjoint copies of $\langle \mathbb{Q}^n, \leqslant_n\rangle$ with order relations between them 
(see Figure~\ref{fig:orderX}). 
The key step is that $\alpha$ is used to define the order-reversing involution  in a way that ensures  the algebra will be an \emph{odd}
Sugihara monoid. 

\eject

Indeed, it is the usage of such a map $\alpha$ 
that distinguishes our construction from other relational representations of similar algebras.
As mentioned above, it 
provides a method for having algebras with ${\sim}1=1$. 
In the case of  representable weakening relation algebras~\cite{GJ20-AU, GJ20-ramics, JS23}, the complement-converse is used as the involutive order-reversing operation on binary relations. It is not difficult to show that no binary relation can be fixed by complement-converse (cf.~\cite[p.46]{Mad10} or ~\cite[p.355]{D82}). We note that in the relational representation of quasi-Boolean algebras by Dunn~\cite{D82}, that construction solves the problem of being able to have an order-reversing involution that fixes certain relations. It does so by taking the complement inside a symmetric relation, which does not have to be $X^2$.

We will view Sugihara monoids as examples of distributive involutive FL-algebras (DInFL-al\-ge\-bras). See Section~\ref{sec:sugihara} for a precise definition of these algebras.
DInFL-algebras have two order-reversing 
operations  
$\sim$ and $-$. If, for all $a \in\mathbf{A}$,  ${\sim} a= -a$ then  the algebra $\mathbf{A}$ is said to be cyclic.   
The use of a non-identity order isomorphism $\alpha$  allows for the construction of concrete \emph{non-cyclic} DInFL-algebras (see item (ii) of Theorem~\ref{thm:InFL-algebra}).  
Our representation of Sugihara monoids is part of a more general method which can also handle non-cyclic DInFL-algebras.

Before giving an outline of the content of this paper, we comment on the recent representation of Sugihara monoids given by Kramer and Maddux~\cite{KM20}. 
That paper provides an alternative solution to the problem of defining an involutive order-reversing operation on relations that has a fixed point. The Sugihara chains are embedded into proper relations algebras (see~\cite[Definition 4]{KM20} for the definition of the binary relations that will be the elements of the Sugihara chain). For an abstract relation algebra, they define \emph{relativized subreducts}~\cite[Definition 8]{KM20} where certain operations from the original relation algebra are retained, while others are replaced by new term-definable operations. 
This notion of relativized subreduct is then implemented on algebras of relations (see~\cite[Theorem 2]{KM20}) in order to obtain a new ${\sim}$ operation that will have a fixed point, hence allowing for the representation of odd Sugihara monoids. 
A difference between our approach and the one in~\cite[Theorem 5]{KM20} is that our representations
are embedded into distributive InFL-algebra of upsets (see Section~\ref{sec:preliminaries}). In their case, the concrete odd Sugihara chain is embedded into a proper relation algebra. A feature of our construction is that the relations which are the elements of the representation are all weakening relations with respect to a partial order. 
	
Our paper is laid out as follows. 
In Section~\ref{sec:preliminaries} we give some basic definitions and background on Sugihara monoids and binary relations. We also outline how we use partially ordered equivalence relations (equipped with an order automorphism $\alpha$) to obtain relational representations of distributive involutive FL-algebras, and give some examples. In Section~\ref{sec:construction} we build the finite odd Sugihara 
chains $\mathbf{S}_{2n+1}$ ($n \in \mathbb{Z}^+$)
as algebras of binary relations. 
The proofs contained in 
 Section~\ref{sec:construction}
 are rather technical, but this is mostly due to the large number of cases that need to be considered. The feature of our construction is the relative simplicity of the relations from which the algebra is built.

In the final section,  we first give a definition of a representable Sugihara monoid. We then show that the class of representable  DInFL-algebras (see Definition~\ref{def:RDInFL}) is closed under ultraproducts (Theorem~\ref{thm:Pu(RDInFL)=RDInFL}). Then, combining  the main result of  Section~\ref{sec:construction} with a representation of finite even Sugihara chains by Maddux~\cite{Mad10}, we are able to show that all Sugihara monoids are representable using the partially ordered equivalence relations described in Section~\ref{subsec:construct}.

\section{Preliminaries}\label{sec:preliminaries}
	
Before we provide the details of our construction, we give some background on residuated lattices, Sugihara monoids, the algebra of binary relations, some necessary results from our previous work~\cite{CR23}, as well as a definition of representability of a DInFL-algebra.  
	
\subsection{Sugihara monoids and DInFL-algebras}\label{sec:sugihara}
	
A \emph{residuated lattice}  is an algebra 
$\mathbf{A}=\langle A, \wedge, \vee, \cdot, 1, \backslash,/\rangle$ such that $\langle A, \wedge, \vee\rangle$ is a lattice, $\langle A, \cdot, 1\rangle$ is a monoid and the following holds for all $a,b,c \in A$: 
$$a \cdot b \leqslant c \quad \Longleftrightarrow \quad a \leqslant c/b
\quad \Longleftrightarrow \quad b \leqslant a\backslash c. $$
If the monoid operation of a residuated lattice $\mathbf{A}$ is commutative, then both residuals $\backslash$ and $/$ (sometimes denoted $\to$ and $\leftarrow$) coincide. If $a \cdot a =a$ is satisfied for all $a \in A$, then $\mathbf{A}$ is an \emph{idempotent} residuated lattice.
A \emph{Sugihara monoid} is an algebra 
$\mathbf{A}=\langle A, \wedge, \vee, \cdot \to, 1, {\sim}\rangle$ such that
$\langle A, \wedge, \vee, \cdot \to,1\rangle$
is a commutative distributive idempotent residuated lattice, and for all $a,b \in A$ the identities ${\sim}{\sim}a=a$ and $a \to {\sim} b=b \to {\sim}a$ are satisfied. (A commutative residuated lattice satisfying those two identities is often called an \emph{involutive} commutative residuated lattice.)

It was shown by Dunn~\cite{AB75} that the variety of Sugihara monoids is generated by algebras that are linearly ordered, and thus the linear ordered algebras are of particular interest.  
For our explicit representations in Section~\ref{sec:construction}, we will only be 
 considering finite Sugihara chains. These are algebras $\mathbf{S}_n$ for $n \in \omega$ and $n \geqslant 2$. If $n=2k$ for $k>0$ then $S_n=\{a_{-k}, \dots, a_{-1}, a_1, \dots, a_k\}$, and if $n=2k+1$ for $k>0$ then 
$S_n=\{a_{-k}, \dots, a_{-1}, a_0, a_1, \dots a_k\}$. The meet and join are defined as expected: $a_i \wedge a_j = a_{\text{min}\{i,j\}}$ and $a_i \vee a_j = a_{\text{max}\{i,j\}}$. The unary operation ${\sim}$ is defined by ${\sim}a_j = a_{-j}$. The monoid operation can be defined as:
	$$a_i \cdot a_j = \begin{cases}
		a_i & \text{if }\, |j|<|i| \\ 
		a_j & \text{if }\, |i|< |j| \\ 
		a_{\text{min}\{i,j\}} & \text{if }\, |j|=|i|.
	\end{cases}$$ 
Lastly, the implication $\to$ is defined by 
$$ 
	a_i \to a_j = \begin{cases}
		{\sim}a_i \vee a_j  & \text{if } i \leqslant j \\
		{\sim}a_i \wedge a_j & \text{if } i >j. 
	\end{cases}
	$$ 
If $n$ is odd then the monoid identity is given by $1=a_0$, and if $n$ is even, then the monoid identity is $1=a_1$. This gives us algebras $\mathbf{S}_n = \langle S_n, \wedge, \vee, \cdot, \to, 1, {\sim}\rangle$.
%cr 20250310
(When we talk of a Sugihara chain, we mean a
linearly ordered Sugihara monoid, using slightly different terminology from that
 used in~\cite{Mad10} and~\cite{KM20} where a Sugihara chain is a linearly ordered Sugihara algebra.)
 
The concrete algebras we describe in Section~\ref{subsec:construct} will be special residuated lattices, which we describe briefly here. 
A \emph{Full Lambek algbebra} (cf.~\cite{GJKO}) is a residuated lattice $\mathbf{L}$ that has an additional constant $0$ (about which nothing is assumed) in the signature. We use the common abbreviation FL-algebra. Using the residuals, two unary operations ${\sim}$ and $-$ can be defined. These are known as \emph{linear negations} and are defined for $a\in L$ by ${\sim}a=a\backslash 0$ and $-a=0/a$. If ${\sim}$ and $-$ satisfy the involutive property, i.e. $-{\sim}a=a={\sim}{-}a$ for all $a \in L$, then $\langle L, \wedge, \vee, \cdot,1,0,\backslash,/\rangle$ is an \emph{InFL-algebra}. When the underlying lattice is distributive, we use the abbreviation DInFL-algebra. If ${\sim}a=-a$ for all $a \in L$ then it is a \emph{cyclic} InFL-algebra. 

We prefer to use the linear negations rather than the residuals in the signature as it is known that such an equivalent formulation is available~\cite[Lemma 2.2]{GJ13}. 
In addition, that equivalent formulation allows us to leave the $0$ out of the signature. 
That is, we will consider DInFL-algebras as algebras of the following signature:
$\mathbf{A}=\langle A, \wedge, \vee, \cdot, 1, {\sim},-\rangle$. Notice that if a Sugihara monoid uses an
alternative signature 
with $- = {\sim}$ then it will be a DInFL-alegbra. Moreover, it will be a commutative cyclic idempotent DInFL-algebra.  

\subsection{Binary relations}\label{sec:bin-rel}
	
	Below we recall some basic definitions and facts about binary relations and the various operations defined on them. 
We will use the following special binary relations on a set $X$: the \emph{empty relation} $\varnothing$, 
	the \emph{identity relation} (\emph{diagonal relation}) $\mathrm{id}_X = \left\{\left(x, x\right) \mid x \in X\right\}$,
	and the \emph{universal relation} $X^2$. The complement of $R$ will be denoted by $R^c$ 
	and the \emph{converse of $R$} is 
	$R^\smile = \left\{\left(x, y\right) \mid \left(y, x\right) \in R\right\}$.
	The \emph{composition} of $R, S \subseteq X^2$ is 
	$R \circ S = \left\{\left(x, y\right) \mid \left(\exists z \in X\right)\left(\left(x, z\right) \in R \textnormal{ and } \left(z, y\right) \in S\right)\right\}$.
	Besides composition, binary relations have the familiar operations of union and intersection that will be used to model join and meet in the relational representations.  
	
	\begin{proposition}\label{prop:properties_binary_relations}
		Let $R, S, T \subseteq X^2$. 
		Then the following hold:
		\begin{multicols}{2}
			\begin{enumerate}[\normalfont (i)]
				\item $R^{\smile\smile} = R$
				\item $R^{\smile c} = R^{c\smile}$
				\item $\left( R \cup S\right)^\smile = R^\smile \cup S^\smile$ 
\item $\left( R \cap S\right)^\smile = R^\smile \cap S^\smile$ 
							\item $\mathrm{id}_X \circ R = R \circ \mathrm{id}_X = R$
				\item $\left(R \circ S\right)\circ T = R \circ \left(S \circ T\right)$
				\item $\left(R\circ S\right)^\smile = S^\smile \circ R^\smile$
				\item $(R \cup S) \circ T = (R \circ T) 
				\cup (S \circ T)$ 
				
				\item $R \circ (S \cup T) = (R \circ S) \cup (R \circ T)$ 
				\end{enumerate}
		\end{multicols}
	\end{proposition}
	
A crucial part of our construction is the use of an additional bijective function $\alpha : X \to X$. We will use $\alpha$ to denote both the function acting on elements of $X$ and the binary relation on $X$. 
The lemma below applies in particular when $\gamma$ is a bijective function from $X$ to $X$.

\begin{lemma}[{\cite[Lemma 3.2]{CR23}}]\label{lem:important_eq_injective_map}
Let $X$ be a set and $R, \gamma \subseteq X^2$. If 
$\gamma$
satisfies 
$\gamma^{\smile}\circ \gamma = \mathrm{id}_X$ and $\gamma \circ \gamma^{\smile}=\mathrm{id}_X$
then the following hold:	
\begin{enumerate}[\normalfont (i)]
\item $\left(\gamma\circ R\right)^c = \gamma \circ R^c$
\item $\left(R\circ \gamma\right)^c = R^c \circ \gamma$
\end{enumerate}	
\end{lemma}

\subsection{Constructing InFL-algebras of relations}\label{subsec:construct}
	
In previous work~\cite[Section 3]{CR23}, we used the construction outlined below to produce concrete algebras of binary relations which formed  distributive quasi relation algebras (DqRAs). These algebras were described by Galatos and Jipsen~\cite{GJ13} and are  algebras of the form $\mathbf{A} = \langle A, \wedge,\vee, \cdot, 1, 0,{\sim},-,'\rangle$. The unary operations ${\sim}, -$ and $'$ are dual lattice isomorphisms satisfying some additional conditions. We do not go into further detail here about the definition of quasi relation algebras as we will not make use of the full signature of these algebras. 
The construction outlined below will produce distributive involutive FL-algebras. 
	
The definitions and results below all involve an equivalence relation $E$ on a set $X$. The special case of $E=X^2$ is vastly simpler and, fortunately, this is the setting used in Section~\ref{sec:construction}. However, we need the more general setting for the representability results in Section~\ref{sec:RSM}.

Let $\mathbf X = \left\langle X, \le\right\rangle$ be a poset with an equivalence relation $E$ satisfying ${\leqslant} \subseteq E$. The set $E$ can be partially ordered for any $(u, v), (x, y) \in E$, as follows:
$$(u, v) \preceq (x, y) \qquad  \textnormal{iff} \qquad x \le u \textnormal{ and } v\le y.$$
	
Then $\mathbf{E}=\left\langle E, \preceq\right\rangle$ is a poset and the set of its upsets, denoted $\mathsf{Up}\left(\mathbf{E}
\right)$, ordered by inclusion, is a distributive lattice. Notice that when $E=X^2$ then $\langle X^2,\preceq\rangle $ is exactly the product poset $\mathbf{X}^{\partial} \times \mathbf{X}$ where the order on $\mathbf{X}^\partial$ is $\geqslant$. 

We note that 
when $E=X^2$
the lattice constructed above is identical to that used by Galatos and Jipsen~\cite{GJ20-AU, GJ20-ramics}, and more recently by Jipsen and \v{S}emrl~\cite{JS23}. In order to represent odd Sugihara monoids as algebras of relations, we will make use of an order automorphism $\alpha$ on the poset $\mathbf{X}$. It is this $\alpha$ that will allow us to define a unary  operation ${\sim}$ on $\mathsf{Up}(\langle X^2,\preceq\rangle )$ and an identity $1\in \mathsf{Up}(\langle X^2,\preceq\rangle )$ such that ${\sim}1=1$.

The lemma below is stated in the original source~\cite{CR23} 
with an additional map $\beta : X \to X$ which is a dual order automorphism. This is not required for the setting of DInFL-algebras. 
For a poset $\mathbf{Y}=\langle Y,\leqslant_Y\rangle $, we denote by $\mathsf{Down}(\mathbf{Y})$ the set of downsets of $\mathbf{Y}$. 
The lemma shows how upsets/downsets of $\mathbf{E}=\langle E,\preceq\rangle $ are produced when composition, converse and complement are applied to upsets/downsets of $\mathbf{E}=\langle E,\preceq\rangle $. It also details what happens when the graph of an order automorphism is composed with an upset. For $R \subseteq X^2$, we will abuse notation and write $R^c$ for the complement of $R$ in $E$. 

\begin{lemma}[{\cite[Lemma 3.5]{CR23}}]\label{lem:important_up-_and_down-sets}
Let $\mathbf X = \langle X, \le\rangle$ be a poset, $E$ an equivalence relation with ${\leqslant} \subseteq E$  and let  $\alpha: X \rightarrow X$ be an order automorphism of $\mathbf X$ with $\alpha \subseteq E$. 		
\begin{enumerate}[\normalfont (i)]
\item If $R, S  \in \mathsf{Up}\left(
\mathbf{E} \right)$, then $R \circ S \in \mathsf{Up}\left( \mathbf{E} \right)$.
\item $R \in \mathsf{Up}\left(
\mathbf{E} \right)$ iff $R^c \in \mathsf{Down}\left(
\mathbf{E} \right)$. 
\item $R \in \mathsf{Up}\left(
\mathbf{E} \right)$ iff $R^{\smile} \in \mathsf{Down}\left(
\mathbf{E} \right)$.
\item If $R \in \mathsf{Up}\left(
\mathbf{E}
\right)$, then $\alpha\circ R \in \mathsf{Up}\left(
\mathbf{E} \right)$ and $R \circ \alpha \in \mathsf{Up}\left( \mathbf{E} \right)$. 
\end{enumerate}	
\end{lemma}
	
The next lemma shows not only that $\leqslant$ is an upset of  $\mathbf{E}$,
but also that it is the identity of composition in $\mathsf{Up}( \mathbf{E} )$. 

\begin{lemma}[{\cite[Lemma 3.6]{CR23}}]\label{lem:order_identity}
Let $\mathbf X = \langle X, \le\rangle$ be a poset, $E$ and equivalence relation with ${\leqslant} \subseteq E$ and $\alpha$ an order automorphism on $X$ with $\alpha \subseteq E$. Then 
\begin{enumerate}[\normalfont (i)]
\item $\le\; \in \mathsf{Up}\left(\mathbf{E}
\right)$;
\item $\le \circ\, R = R\; \circ \le\; = R$ for all $R \in \mathsf{Up}\left(\mathbf{E}
\right)$. 
\end{enumerate}	
\end{lemma}

Consider a poset $\mathbf X = \langle X, \le\rangle$ with $E$ and $\alpha$ as used above. 
Relational composition is associative and, by Lemma~\ref{lem:order_identity},  $\le$ is the identity of composition. Hence the structure $\left\langle\mathsf{Up}\left(
\mathbf{E}
\right),\circ, \le\right\rangle$ is a monoid. 
	In addition, composition is residuated and the residuals $\backslash, /$ are defined by the usual expressions for residuals on binary relations: $R\backslash S = (R^{\smile}\circ S^c)^c$ and $R/S = (R^c\circ S^{\smile})^c$. The algebra $\left\langle \mathsf{Up}\left(
\mathbf{E} 
 \right), \cap, \cup, \circ, \leqslant, \backslash, /, \right\rangle$ is therefore a (distributive) residuated lattice.
	
The next lemma shows that if $\alpha$ is an order automorphism of a poset $\langle X, \le\rangle$ with an equivalence relation $E $ such that ${\leqslant} \subseteq E$, then $\alpha$ can be used to define an element of $\mathsf{Up}(\mathbf{E})$ to act as a $0$ that will give rise to an InFL-algebra.  
	
\begin{lemma}[{\cite[Lemma 3.9]{CR23}}]\label{lem:definition_of_0}
Let $\mathbf X = \langle X, \le\rangle$ be a poset with an equivalence relation $E$ on $X$ satisying ${\leqslant} \subseteq E$. Further, let $\alpha: X \rightarrow X$ be an order automorphism of $\mathbf X$ with $\alpha \subseteq E$. Then $\alpha \,\circ \le^{c\smile} =$ $ \le^{c\smile} \circ \,\alpha \in \mathsf{Up}\left(\mathbf{E}\right)$.
	\end{lemma}

We will now add $0 = \alpha\,\circ \le^{c\smile} \,= \,\le^{c\smile}\circ \,\alpha$ to the signature of the distributive residuated   lattice 
 $\left\langle \mathsf{Up}\left(
\mathbf{E}
\right), \cap, \cup, \circ, \leqslant, \backslash, /\right\rangle$  to obtain the FL-algebra
$\left\langle \mathsf{Up}\left(
\mathbf{E}
\right), \cap, \cup, \circ, \le, 0,\backslash, / \right\rangle$. 
The next lemma gives an alternative way of calculating ${\sim} R$ and $- R$ if we set $0 = \alpha\, \circ \le^{c\smile} \, = \,\le^{c\smile} \circ \,\alpha$. 
	
\begin{lemma}[{\cite[Lemma 3.10]{CR23}}]\label{lem:twiddle_minus_alternative}
Let $\mathbf X = \langle X, \le\rangle$ be a poset with $E$ an equivalence relation and $\alpha: X \rightarrow X$ an order automorphism of $\mathbf X$ where ${\leqslant} \subseteq E$ and $\alpha \subseteq E$.  
If  $0 = \alpha\, \circ \le^{c\smile} \, = \,\le^{c\smile} \circ \,\alpha$ 
then 
${\sim} R = R^{c\smile}\circ \alpha$ and $-R = \alpha \circ R^{c\smile}$ for all $R \in \mathsf{Up}\left(\langle X^2, \preceq\rangle\right)$.
\end{lemma}

Using the alternative method of calculating ${\sim}$ and ${-}$ from Lemma~\ref{lem:twiddle_minus_alternative}, greatly simplifies the proof of the main theorem of this background section. 

\begin{theorem}[{\cite[Theorem 3.12]{CR23}}]\label{thm:InFL-algebra}
Let $\mathbf X = \langle X, \le\rangle$ be a poset with $E$ an equivalence relation on $X$ such that ${\leqslant} \subseteq E$. Further, let $\alpha: X \rightarrow X$ be an order automorphism of $\mathbf X$ such that $\alpha \subseteq E$.  
\begin{enumerate}[\normalfont (i)]
\item If\, $0 = {\alpha}\, \circ \,{\leqslant^{c\smile}} = {\leqslant^{c\smile}} \circ {\alpha}$ then $\mathfrak{D}(\mathbf{E})=\langle \mathsf{Up}(\mathbf{E}), \cap, \cup, \circ, \leqslant, {\sim},-\rangle$ is a distributive InFL-algebra. 
\item The algebra in {\upshape (i)} is cyclic if and only if $\alpha=\mathrm{id}_X$. 		\end{enumerate}
\end{theorem} 
We define the class of \emph{equivalence distributive involutive FL-algebras}, denoted $\mathsf{EDInFL}$, to be those algebras of the form $\mathfrak{D}(\mathbf{E})$ from part (i) of Theorem~\ref{thm:InFL-algebra}. If $E=X^2$ then we say that $\mathfrak{D}(\mathbf{E})$ is a \emph{full distributive involutive FL-algebra}, and denote the class of such algebras by $\mathsf{FDInFL}$.

Before we give some examples, we need the definition of a \emph{direct reduct}. This is an adaptation of~\cite[Definition 8]{KM20}. It differs from an ordinary reduct because the set of operations is not simply a subset of the original, but new operations are also defined in terms of the old operations.

\begin{definition}\label{def:direct}
Let $\mathbf{A} =\langle 
A, \wedge, \vee, \cdot, 1, {\sim},-\rangle$ be a DInFL-algebra. We call 
$\mathbf{A}_r=\langle A, \wedge, \vee, \cdot, \to, 1, {\sim}\rangle$
the \emph{direct reduct} of $\mathbf{A}$, where $\to$ is defined for $a,b \in A$ by 
$a\to b = {\sim}(-b \cdot a)$. 
\end{definition}

If a 
DInFL-algebra $\mathbf{A} =\langle 
A, \wedge, \vee, \cdot, 1, {\sim},-\rangle$
is commutative, then it is cyclic. Hence it is easy to see that if $\mathbf{A}$ is commutative and idempotent, then the direct reduct of $\mathbf{A}$ is a Sugihara monoid. If in addition, $\mathbf{A}$ is a chain, then $a \to b = {\sim}a \vee b$ when $a \leqslant b$ and $a \to b = {\sim}a \wedge b$ when $a > b$. 

In the concrete case, 
the operation $\Rightarrow$ can be added to $\mathfrak{D}(\mathbf{E})$ using
the following standard definition of 
the left residual of composition: for $R,S \in \mathsf{Up}(\mathbf{E})$ let 
$R \Rightarrow S= \left(R^{\smile}\circ S^c\right)^c
$. Using Lemma~\ref{lem:twiddle_minus_alternative} it is easy to see that this coincides with the expression given in Definition~\ref{def:direct}. 

\begin{example}\label{ex:S2}
The two-element Sugihara chain $\mathbf{S}_2$ can be represented over $X = \{x\}$ with ${\leqslant} = E = \alpha= \mathrm{id}_X$.
\end{example}

Using Theorem~\ref{thm:InFL-algebra}, we obtain a relational representation of the 
odd Sugihara chain $\mathbf{S}_3$. We observe that this representation is not covered by the construction given in Section~\ref{sec:construction}. 
 
\begin{example}\label{ex:S3}
Consider the three-element Sugihara chain $\mathbf{S}_3$. The underlying set is given by $S_3=\{a_{-1},a_0,a_1\}$ where $a_0$ is the identity of the monoid operation. Consider the two-element poset $X = \left\{x, y\right\}$ with order $\leqslant\, = \mathrm{id}_X$. 
We let $E=X^2$ and define  
the order automorphism $\alpha = \left\{\left(x, y\right), \left(y, x\right)\right\}$. We get that $\mathsf{Up}(\langle X^2,\preceq\rangle)$ has as its underlying lattice the 16-element Boolean lattice. It is straightforward to  calculate that ${\leqslant} = {\sim}{\leqslant}  = {\leqslant}^{c\smile} \circ \,\alpha$ and then to show that the set of binary relations $\{\varnothing, \leqslant, X^2\}$ 
is closed under $\cap$, $\cup$, $\circ$, $\Rightarrow$ and ${\sim}$. 
Define a map $h: S_3 \to \mathsf{Up}(\langle X^2,\preceq\rangle)$ by $h(a_{-1})=\varnothing$, $h(a_0)={\leqslant}$ and $h(a_1)=X^2$. It is easy to check that this is indeed an embedding of $\mathbf{S}_3$ into the direct reduct of $\mathfrak{D}(\mathbf{E})$. 
\end{example}

The next example is based on the representation of the finite even Sugihara chains $\mathbf{S}_{2n+2}$ ($n \in \mathbb{Z}^+$) by Maddux~\cite[Theorem 6.2]{Mad10}. The exposition below shows how that representation can be re-cast into our setting.  We will need this example in Section~\ref{subsec:RSM=SM}. 

\begin{example}\label{ex:Maddux2010}
Let $n \in \mathbb{Z}^+$ and consider $\mathbb{Q}^n = \left\{\left(q_1, \ldots, q_n\right) \mid q_1, \ldots, q_n \in \mathbb{Q}\right\}$. We will denote the $n$-tuple $\left( q_1, \dots, q_n\right)$ by $q$. Now define  binary relations $\mathrm{id}_{\mathbb{Q}^n}$ and $L_1$ on $\mathbb{Q}^n$ as follows: 
\begin{align*}
\left(p, q\right) \in \mathrm{id}_{\mathbb{Q}^n} \qquad & \textnormal{ iff } \qquad p = q \\
(p, q) \in L_1 \qquad & \textnormal{ iff } \qquad p_1 < q_1.
\end{align*}
For $i \in \{2, \ldots, n\}$, the relation  $L_i$ on $\mathbb{Q}^n$ is defined follows:
$$
(p, q) \in L_i \qquad \textnormal{ iff } \qquad \left(p_1, \ldots, p_{i-1}\right) = \left(q_1, \ldots, q_{i-1}\right) \textnormal{ and } p_i < q_i
$$
The relations in $\left\{\mathrm{id}_{\mathbb{Q}^n}, L_1, \ldots, L_n\right\}$ are pairwise disjoint.
We set $\leqslant_n := \bigcup\left\{\mathrm{id}_{\mathbb{Q}^n}, L_1, \ldots, L_n\right\}$. 
Let $\mathbf{X}=
\langle \mathbb{Q}^n,\leqslant_n\rangle$, $E=\left(\mathbb{Q}^n\right)^2$ and $\alpha=\mathrm{id}_{\mathbb{Q}^n}$. Notice that with $\alpha=\mathrm{id}_{\mathbb{Q}^n}$ we get ${\sim} \leqslant_n=\left(\leqslant_n\right)^{c\smile}= {<_n}$. Moreover, by Theorem~\ref{thm:InFL-algebra}, the
algebra $\mathfrak{D}(\mathbf{E})$ is a cyclic DInFL-algebra.

For $i \in \left\{-n-1, \ldots, -1, 1, \ldots, n+ 1\right\}$, define 
$$
T_i = \begin{cases} 
\varnothing & \text{if }\, i = -n-1 \\ 
\displaystyle \bigcup_{j=1}^{n+1+i}L_j & \text{if }\, -n \leqslant i \leqslant -1 \\ 
\leqslant_n & \text{if }\, i = 1\\
\left(T_{-i}\right)^{c\smile} & \text{if } \, 2\leqslant i \leqslant n+1
\end{cases}
$$
These relations are upsets of $\langle E, \preceq\rangle $ and form a Sugihara chain of length $2n+2$, i.e., 
$$\varnothing=T_{-n -1} \subseteq T_{-n} \subseteq \cdots \subseteq T_{-1} \subseteq T_1 \subseteq \cdots \subseteq T_{n} \subseteq T_{n+1} = \left(\mathbb{Q}^n\right)^2.$$
More precisely, the even Sugihara chain $\mathbf{S}_{2n+2}$ can be embedded into the direct reduct of the algebra $\mathfrak{D}(\mathbf{E})$ by setting $\varphi\left(a_i\right) = T_i$ for each $i \in \left\{-n-1, \ldots, -1, 1, \ldots, n+ 1\right\}$.  
\end{example}

\subsection{Representable DInFL-algebras} \label{sec:RDInFL}

Here we recall definitions and results from~\cite[Section 4]{CR23}. These will be needed in Section~\ref{sec:RSM}. As before, the major difference is that the results there were proved for distributive quasi relation algebras, but here we restrict to the simpler case of distributive InFL-algebras. 

Suppose we have  an index set $I$ and for each $i \in I$, let $\mathbf X_i  = \langle X_i, \le_i\rangle $ be a poset and $E_i$ an equivalence relation with  $\le_i\; \subseteq E_i$. 
We further require that $X_i \cap X_j = \varnothing$ whenever $i \neq j$.
Then set 
$$
X := \displaystyle \bigcup_{i\in I}X_i, \quad \displaystyle\le \; :=\bigcup_{i\in I}\le_i \quad \textnormal{and} \quad E := \displaystyle \bigcup_{i \in I}E_i.
$$
If we have maps $\gamma_i: X_i\rightarrow X_i$ ($i \in I$) define $\gamma: X \to X$ by setting, for each $x \in X$, $\gamma(x) = \gamma_i(x)$ whenever $x \in X_i$.
Clearly $\gamma = \displaystyle\bigcup_{i\in I}\gamma_i$. 
These definitions will be used throughout the theorems and their proofs in this section.

\begin{theorem}\cite[Theorem 4.1]{CR23}
\label{thm:prodXEalpha} 
Let  $I$ be an index set. For each $i \in I$, let $\mathbf X_i  = \langle X_i, \le_i\rangle $ be a poset and $E_i$ an equivalence relation such that $\le_i\; \subseteq E_i$. Assume  $X_i \cap X_j = \varnothing$ for all $i \neq j$. For each $i \in I$ we have $\alpha_i: X_i \rightarrow X_i$ an order automorphism of $\mathbf X_i$. Then:
\begin{enumerate}[\normalfont (i)]
\item $\mathbf{X}=\langle X, \le\rangle$ is a poset and $E$ is an equivalence relation on $X$ such that $\le\; \subseteq E$;
\item $\alpha$ is an order automorphism of $\mathbf X$ such that $\alpha \subseteq E$;
\item for the DInFL-algebras
$\mathfrak{D}\left(\mathbf E\right)$ and 
$\displaystyle \prod_{i \in I}\mathfrak{D}\left(\mathbf E_i\right)$, we have 
$\mathfrak{D}\left(\mathbf E\right) \cong \displaystyle \prod_{i \in I}\mathfrak{D}\left(\mathbf E_i\right)$. 
\end{enumerate}
\end{theorem}

The following theorem relies on part (iii) of Theorem~\ref{thm:prodXEalpha} and is used to define representable DInFL-algebras. 

\begin{theorem}\cite[Theorem 4.4]{CR23}
$\mathbb{P}\left(\mathsf{FDInFL}\right) = \mathsf{EDInFL}$.
\end{theorem}

The definition below is a generalisation of~\cite[Definition 4.5]{CR23}, given there for DqRAs. 
\begin{definition}\label{def:RDInFL}
A DInFL-algebra  $\mathbf{A} = \left\langle A, \wedge, \vee, \cdot, 1,  {\sim},{-}\right\rangle$ is \emph{representable} if 
$\mathbf{A} \in \mathbb{ISP}\left(\mathsf{FDInFL}\right)$
or, equivalently, $\mathbf{A} \in \mathbb{IS}\left(\mathsf{EDInFL}\right)$. 
\end{definition}

We denote the class of representable DInFL-algebras by $\mathsf{RDInFL}$. 
In Theorem~\ref{thm:Pu(RDInFL)=RDInFL} we will show that 
$\mathsf{RDInFL}$
is closed under ultraproducts. That will be a key step in showing that all Sugihara monoids can be represented as algebras of binary relations. 

We end this section by observing that if one considers posets $\langle X, = \rangle$ (i.e. with the discrete order), and $\alpha=\mathrm{id}_X$, then the constructions and definitions above, with some slight modification, will coincide with those for relation algebras. The modifications include using converse, complement, and bounds in the signature. 

\section{Constructing  
finite 
odd Sugihara 
chains}\label{sec:construction}

Let $\langle C, \leqslant\rangle$ be a chain. Set $X := \left(C\times \left\{-1\right\}\right) \cup \left(C\times \left\{1\right\}\right)$. If $b \in \left\{-1, 1\right\}$, we will write $x^b$ instead of $\left(x, b\right)$. Define a binary relation $\leqslant_X$ on $X$ by
\[
x^b \leqslant_X y^d \quad \text{ iff } \quad x^b = y^d  \quad \textnormal{ or } \quad x < y 
\]

It is not difficult to see that $\langle X,\leqslant_X\rangle$ is a poset. We can define a map $\alpha: X \to X$ by setting, for each $x^b \in X$,  $\alpha(x^b) = x^{-b}$. The map $\alpha$ will be a self-inverse order automorphism of $\left\langle X, \leqslant_X\right\rangle$. This poset is drawn in Figure~\ref{fig:orderX}. 

%20231019ac figure added 
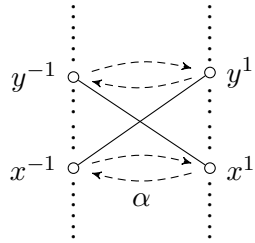
\begin{figure}[th]
\vspace{-0.4cm}
\centering
\begin{tikzpicture}[scale=0.6]
\begin{scope}
  % Elements
  \node[unshaded] (fn) at (-1,1) {};
  \node[unshaded] (tn) at (2,1) {};
  \node at (-1,0.6) {$\vdots$};
  \node at (-1,-0.1) {$\vdots$};
  \node at (-1,3.8) {$\vdots$};
  \node at (-1,4.5) {$\vdots$};

  \node at (-1,1.85) {$\vdots$};
  \node at (-1,2.65) {$\vdots$};

  \node at (2,0.6) {$\vdots$};
  \node at (2,-0.1) {$\vdots$};
  \node at (2,3.8) {$\vdots$};
  \node at (2,4.5) {$\vdots$};

  \node at (2,1.85) {$\vdots$};
  \node at (2,2.65) {$\vdots$};
  \node[unshaded] (f1) at (-1,3) {};
  \node[unshaded] (t1) at (2,3.1) {};
  % Order
  \draw[order] (fn) -- (t1);
  \draw[order] (tn) -- (f1);
  % Labels
  \path (fn) edge [map, bend left=15] node {} (tn);
  \path (tn) edge [map, bend left=15] node {} (fn);
  \path (f1) edge [map, bend left=15] node {} (t1);
  \path (t1) edge [map, bend left=15] node {} (f1);
  \node[label] at (0.5,0.3) {$\alpha$};

  \node[label,anchor=east] at (fn) {$x^{-1}$};
  \node[label,anchor=west] at (tn) {$x^1$};
  \node[label,anchor=east] at (f1) {$y^{-1}$};
  \node[label,anchor=west] at (t1) {$y^1$};
  \end{scope}
\end{tikzpicture}

\caption{The poset $\langle X,\leqslant_X\rangle$ from a chain $\langle C,\leqslant\rangle$ with order automorphism $\alpha: X \to X $. }\label{fig:orderX}
\end{figure}

The first lemma in this section establishes that when we have a poset $\langle X,\leqslant_X\rangle$ constructed from a chain as above, the order $\leqslant_X$ will be fixed by the ${\sim}$ operation described in Theorem~\ref{thm:InFL-algebra}.
Note that $(x^b,y^d) \in {\sim}R$ if and only if $(y^{-d},x^b) \notin R$. In particular, $(x^b,y^d) \in 0$ if and only if $y^{-d} \nleqslant x^b$. 

\begin{lemma}\label{lem:leq=0} 
Let $\langle C,\leqslant\rangle$ be a chain, and $\langle X,\leqslant_X\rangle$ with $\alpha$ as above. Then 
$\leqslant_X\, = 0 = (\leqslant_X)^{c\smile}\circ \alpha$.  
\end{lemma}
\begin{proof}
To prove the left-to-right inclusion, assume $x^b \leqslant_X y^d$. Then $x^b = y^d$ or $x < y$. If $x^b = y^d$, then $y^{-d} \not\leqslant x^{b}$, and so $\left(x^b, y^d\right) \in   0$. 
On the other hand, if $x < y$, then $y \not\leqslant x$, so $y^{-d} \not\leqslant x^{b}$. 
Again we have $(x^b, y^d) \in 0$.
For the other inclusion, let $\left(x^b, y^d\right) \in 0 =  \left(\leqslant_X\right)^{c\smile}\circ \alpha$. Then we have 
$y^{-d} \not\leqslant x^{b}$. If $b = -d$, then 
$y \not\leqslant x$, so
$x < y$. Hence, $x^b \leqslant_X y^d$. On the other hand, if $b \neq -d$, then $y^{-d} \not\leqslant x^{b}$ implies $y \not < x$. This means $x \leqslant y$, and so, since $d = b$, we get $x^b \leqslant_X y^d$. 
\end{proof}

Consider the 
rationals  $\left\langle \mathbb{Q}, \leqslant\right\rangle$ with their usual order.
Let $\mathbb{Q}^n = \left\{\left(q_1, \ldots, q_n\right)\mid q_1, \ldots, q_n \in \mathbb{Q}\right\}$.
We will again denote the $n$-tuple $\left( q_1, \dots, q_n\right)$ by $q$. 
Define binary relations $\mathrm{id}_{\mathbb{Q}^n}$ and $L_{1}$ on $\mathbb{Q}^n$ by 
\begin{align*}
\left(p, q\right) \in \mathrm{id}_{\mathbb{Q}^n} & \quad \textnormal{ iff }\quad  p =  q\\ 
\left(p, q\right) \in L_1 & \quad \textnormal{ iff } \quad p_1 < q_1 
\end{align*}
and for $i \in \left\{2, \ldots, n\right\}$, define binary relations $L_i$ on $\mathbb{Q}^n$ by
\begin{align*}
\left(p, q\right) \in L_i & \quad \textnormal{ iff } \quad \left(p_1, \ldots, p_{i-1}\right) = \left(q_1, \ldots, q_{i-1}\right) \textnormal{ and } p_i < q_i.
\end{align*}
Now set $\leqslant_n := \bigcup\left\{\mathrm{id}_{\mathbb{Q}^n}, L_1, \ldots, L_n\right\}$.  
Then $\leqslant_n$ is the lexicographic order on $\mathbb{Q}^n$. Since the rational numbers is linearly ordered, it follows that $\left\langle \mathbb{Q}^n, \leqslant_n\right\rangle$ is a chain. Part (i) of the lemma below can be used to show that $\leqslant_n$ is transitive. It will also help to simplify some of the proofs later on.

\begin{lemma}\label{lem:transitivity_Tis}
Let $n \geqslant 1$. 
\begin{enumerate}[\normalfont (i)]
\item If $\left(p, q\right) \in L_j$ and $\left(q, r\right) \in L_k$ with $j, k \in \{1, \ldots, n\}$, then $\left(p, r\right) \in L_{\textnormal{min}\{j, k\}}$.
\item The structure $\left\langle \mathbb{Q}^n, \leqslant_n\right\rangle$ is a chain.
\end{enumerate}
\end{lemma}

Let $X := \left(\mathbb{Q}^n\times \left\{-1\right\}\right) \cup \left( \mathbb{Q}^n\times \left\{1\right\}\right)$, 
and define a binary relation $\leqslant_X$ on $X$ by
\[
p^b \leqslant_X q^d \quad \text{ iff } \quad p^b = q^d  \quad \textnormal{ or } \quad p <_n q 
\] 
Consider the posets $\left\langle X, \le_X\right\rangle$ and $\left\langle X^2, \preceq\right\rangle$ (see Section~\ref{subsec:construct}).
 Let $\alpha: X \to X$ be defined by $\alpha\left(p^b\right) = p^{-b}$. For each $1 \leqslant j \leqslant n$, set $$U_j = \left\{\left(p^b, q^d\right)\Big |\; b, d \in \left\{-1, 1\right\} \text{ and } \left(p, q\right) \in L_j\right\}.$$
 Using $\alpha$ and $U_j$, we alter the definition of the binary relation $T_i$ in~\cite[Theorem 6.2]{Mad10} and Example~\ref{ex:Maddux2010} as follows. For $i \in \left\{-n-1, \ldots, n+ 1\right\}$, define 
$$
R_i = \begin{cases} 
\varnothing & \text{if }\, i = -n-1 \\ 
\displaystyle \bigcup_{j=1}^{n+1+i}U_j & \text{if }\, -n \leqslant i \leqslant -1 \\ 
\leqslant_X & \text{if }\, i = 0\\

{\sim} R_{-i} 
= \left(R_{-i}\right)^{c\smile}\circ \alpha & \text{if } \, 1\leqslant i \leqslant n+1
\end{cases}
$$
The relations $R_i$ will be the elements of our representation of the odd Sugihara monoid $\mathbf{S}_{2n+3}$. 

To see how the construction of Maddux~\cite{Mad10} (see Example~\ref{ex:Maddux2010}) relates to our construction, consider the map $\delta$ that sends relations $R$ on $\left(\mathbb{Q}^n\right)^2$ to relations on $X$ as follows:
$$
\delta\left(R\right) = \left\{\left(p^b, q^d\right) \mid b, d \in \{-1, 1\} \textnormal{ and } \left(p, q\right) \in R\right\}.
$$

Clearly $\delta\left(L_j\right) = U_j$. 
It is straightforward to show that $$\delta\left(\bigcup_{j=1}^{n+1+i}L_j\right) = \bigcup_{j=1}^{n+1+i}\delta\left(L_j\right)$$
which means $\delta\left(T_i\right) = R_i$ for all $i \in \left\{-n-1, \ldots, -1\right\}$. Moreover, we can also show that $\delta\left(\left(T_{-i}\right)^{c\smile}\right) = \left(\delta\left(T_{-i}\right)\right)^{c\smile}$ for all $i \in \left\{1, \ldots, n+1\right\}$, which implies that 
$\delta\left(T_i\right) \circ \alpha = R_i$ for all $i \in \left\{1, \ldots, n+ 1\right\}$. 

\begin{lemma}\label{lem:Ri_upsets}
For each $i \in \left\{-n-1, \ldots, n+ 1\right\}$, the relation $R_{i}$ is an upset of $\left\langle X^2, \preceq\right\rangle$.
\end{lemma}

\begin{proof}
Both $X^2$ and $\varnothing$ are upsets of $\left\langle X^2, \preceq\right\rangle$. The fact that $\leqslant_X$ is an upset of $\left\langle X^2, \preceq\right\rangle$ follows from Lemma~\ref{lem:order_identity}. 
It would also follow from Lemmas \ref{lem:important_up-_and_down-sets} and \ref{lem:twiddle_minus_alternative} that $R_i$ is an upset for $i \in \left\{1, \ldots, n\right\}$, if we can show that $R_i$ is an upset of $\left\langle X^2, \preceq\right\rangle$ when $-n \leqslant i < 0$. 

Let $(p^b, q^d) \in R_{i}$ and assume $\left(p^b, q^d\right) \preceq \left(r^e, s^f\right)$. The former implies that $(p^e, q^f) \in U_{j}$ for some $j$ such that $1 \leqslant j \leqslant n + 1+ i$. Hence, $(p, q) \in L_j$ for some $j$ such that $1 \leqslant j \leqslant n + 1+ i$. If $j = 1$, then we have $p_1 < q_1$.  But $\left(p^b, q^d\right) \preceq \left(r^e, s^f\right)$ implies that $r \leq_n p$ and $q \leq_n s$, so $r_1 \leqslant p_1 < q_1 \leqslant s_1$. Hence, $\left(r^e, s^f\right) \in U_1 \subseteq R_{i}$. 

On the other hand, if $1 < j \leqslant n + 1 + i$, then $\left( p_1, \ldots, p_{j-1}\right) = \left( q_1, \ldots, q_{j-1}\right)$ and $p_j < q_j$.  We now consider the following cases:
\\

\noindent
\underline{Case 1:} $r = p$ and $q = s$. In this case we get $\left(r_1, \ldots, r_{j-1}\right) = \left(p_1, \ldots, p_{j-1}\right) = \left(q_1, \ldots, q_{j-1}\right) = \left(s_1, \ldots, s_{j-1}\right)$ and $r_j = p_j < q_j = s_j$. Hence, $\left(r^e, s^f\right) \in U_j \subseteq R_{i}$. 
\\

\noindent
\underline{Case 2:} $r_1 < p_1$ or $q_1 < s_1$. Then $r_1 < p_1 = q_1 \leqslant s_1$ or $r_1 \leqslant  p_1 = q_1 < s_1$. In both cases, we have $\left(r^e, s^f\right) \in U_1 \subseteq R_{i}$.
\\

\noindent
\underline{Case 3:} $r = p$ and $\left(q, s\right) \in L_k$ for some $k \in \left\{2, \ldots, n\right\}$. Here we have $\left(p, s\right) \in L_{\textnormal{min}\{j, k\}}$ by Lemma~\ref{lem:transitivity_Tis}. Hence, since $r = p$, it follows that $\left(r_1, \ldots, r_{\textnormal{min}\{j, k\}-1}\right) = \left(p_1, \ldots, p_{\textnormal{min}\{j, k\}-1}\right)$ and $r_{\textnormal{min}\{j, k\}} = p_{\textnormal{min}\{j, k\}} < s_{\textnormal{min}\{j, k\}}$. This means $\left(r, s\right) \in L_{\textnormal{min}\{j, k\}}$, and so $\left(r^e, s^f\right) \in U_{\text{min}\{j, k\}} \subseteq R_{i}$. \\

\noindent
\underline{Case 4:} $q = s$ and $\left(r, p\right) \in L_k$ for some $k \in \left\{2, \ldots, n\right\}$.  This case is similar to Case 3.
\\

\noindent
\underline{Case 5:} $\left(r, p\right) \in L_k$ and $\left(q, s\right) \in L_\ell$ for $k, \ell \in \left\{2, \ldots, n\right\}$. By Lemma~\ref{lem:transitivity_Tis}, $\left(r, q\right) \in L_{\text{min}\{j, k\}}$, and so, by another application of Lemma~\ref{lem:transitivity_Tis}, we have $\left(r, s\right) \in L_{\text{min}\{\text{min}\{j, k\}, \ell\}}$. Hence, $\left(r, s\right) \in L_{\text{min}\{j, k, \ell\}}$, which means $\left(r^e, s^f\right) \in U_{\text{min}\{j, k, \ell\}} \subseteq R_{i}$.
\end{proof}

Notice that when defining $R_i$ for $i \in \left\{1, \ldots, n+1\right\}$ we chose to use the unary order-reversing operation $\sim$ instead of $-$. Since $\alpha\neq \text{id}_{X}$ we know from Theorem~\ref{thm:InFL-algebra} that the algebra will not be cyclic and hence ${\sim}R\neq -R$ in general. However, the following lemma shows that for the relations $R_i$, we will in fact get ${\sim}R_i=-R_i$. 

\begin{lemma}
For each $i \in \left\{-n-1, \ldots, n+1\right\}$, we have $-R_i = {\sim}R_i$. 
\end{lemma}

\begin{proof}
If $i = 0$, the claim follows from 	Lemma \ref{lem:definition_of_0}.  Now consider the following cases:
\\

\noindent
\underline{Case 1:} $-n-1 \leqslant i < 0$. Let $\left(p^b, q^d\right) \notin - R_i$. By Lemma \ref{lem:important_eq_injective_map}, $\left(p^b, q^d\right) \in \left(-R_i\right)^c = \left(\alpha \circ \left(R_i\right)^{c\smile}\right)^c = \alpha\;\circ \; R^\smile_i$. Therefore, $\left(q^d, p^{-b}\right) \in R_i$, and so  $\left(q^d, p^{-b}\right) \in U_k$ for some $k \in \left\{1, \ldots, n+ 1+i\right\}$. This implies $(q, p) \in L_k$ for some $k \in \left\{1, \ldots, n+ 1+i\right\}$. It thus follows that $\left(q^{-d}, p^{b}\right) \in U_k\subseteq R_i$, i.e., $\left(\alpha\left(q^d\right), p^b\right) \in U_k \subseteq R_i$. Hence, $\left(p^b, \alpha\left(q^d\right)\right) \in R_i^\smile$, and therefore $\left(p^b, \alpha\left(q^d\right)\right) \in \left(\left(R_i\right)^{c\smile}\right)^c$. From this we get $\left(p^b, q^d\right) \in \left(\left(R_i\right)^{c\smile}\right)^c\circ \alpha = \left(\left(R_i\right)^{c\smile}\circ \alpha\right)^c= \left({\sim}R_i\right)^c$, i.e., $\left(p^b, q^d\right) \notin {\sim}R_i$. 

The reverse inclusion can be proved in a similar way. \\

\noindent
\underline{Case 2:} $0 < i \leqslant n+1$. This case follows from the previous case and the fact that $-$ and $\sim$ satisfy the involutive property. 
\end{proof}

We want our set of relations to have the structure of a Sugihara chain, therefore we need to confirm that they form a linear order when ordered by inclusion. 

\begin{lemma}\label{lem:Ri-chain}  
The relations in $\left\{R_{-n -1}, R_{-n}, \ldots,  R_{-1}, R_0, R_1, \ldots, R_{n}, R_{n+1}\right\}$ form a chain, i.e., $$R_{-n -1} \subseteq R_{-n} \subseteq \cdots \subseteq R_{-1} \subseteq R_0 \subseteq R_1 \subseteq \cdots \subseteq R_{n} \subseteq R_{n+1}.$$
\end{lemma}
\begin{proof}
It follows from the construction that $R_{-i} \subseteq R_{-i+ 1}$ for all $i \in \left\{2, \ldots, n + 1\right\}$. Since $\sim$ is order-reversing, it follows that $R_i \subseteq R_{i+1}$ for all $i \in \left\{1, \ldots, n\right\}$. 
To see that $R_{-1} \subseteq R_0 = \;\leqslant_X$, let $\left(p^b, q^d\right) \in R_{-1}$. Then $\left(p^b, q^d\right) \in U_j$ for some $j \in \left\{1, \ldots, n\right\}$, and therefore $(p, q) \in L_j$ for some  $j \in \left\{1, \ldots , n\right\}$. Hence, $p <_n q$, which means $p^b \leqslant_X q^d$. 
The inclusion $R_0 \subseteq R_1$ then follows from the fact that $\sim$ is order-reversing. 
\end{proof}

In $\delta$-notation, the above result can be rewritten as follows:  
$$
\delta\left(T_{-n-1}\right) \subseteq \delta\left(T_{-n}\right) \subseteq  \cdots \subseteq \delta\left(T_{-1}\right) \subseteq R_0 \subseteq \delta\left(T_1\right) \circ \alpha \subseteq \cdots \subseteq \delta\left(T_n\right) \circ \alpha \subseteq \delta\left(T_{n+1}\right) \circ \alpha
$$

The transitivity of each of the relations $R_i$ will be essential for the calculations involving relational composition. 

\begin{lemma}\label{lem:Ri_transitive}
For each $i \in \left\{-n-1, \ldots, n+ 1\right\}$, the relation $R_i$ is transitive. 
\end{lemma}

\begin{proof}
This is clearly true for $i=0$ and $i = n+1$. Since $\varnothing$ is vacuously transitive, the statement also holds for $i = -n-1$.  We now consider the following two cases:
\\

\noindent
\underline{Case 1:} $ -n \leqslant i < 0$. Let $\left(p^b, r^e\right) \in R_i$ and $\left(r^e, q^d\right) \in R_i$. Then there exist $j, k \in \left\{1, \ldots, n+1 +i\right\}$ such that $(p, r) \in L_j$ and $(r, q) \in L_k$. Hence, by Lemma~\ref{lem:transitivity_Tis}, we have $\left(p, q\right) \in L_{\text{min}\{j, k\}}$, and so $\left(p^b, q^d\right) \in U_{\text{min}\{j, k\}} \subseteq R_i$.\\

\noindent
\underline{Case 2:} $0 < i \leqslant n$. Let $\left(p^b, r^e\right) \in R_i$ and $\left(r^e, q^d\right) \in R_i$. Then we have $\left(p^b, r^e\right) \in {\sim}R_{-i}$ and $\left(r^e, q^d\right) \in {\sim}R_{-i}$. Hence, $\left(r^{-e}, p^b\right) \notin R_{-i}$ and $\left(q^{-d}, r^e\right) \notin R_{-i}$, which means $(r, p) \notin L_j$ for all $j \in  \left\{1, \ldots, n+1 -i\right\}$ and $(q, r) \notin L_k$ for all $k \in \left\{1, \ldots, n+1 -i\right\}$. The first part implies that $r_1 \not < p_1$ and, for all $j \in \left\{2, \ldots, n+1 -i\right\}$, if $\left(r_1, \ldots, r_{j-1}\right) = \left(p_1, \ldots, p_{j-1}\right)$ then $r_j \not < p_j$, while the second part gives $q_1 \not < r_1$ and, for all $k \in \left\{2, \ldots, n+1 -i\right\}$, if $\left(q_1, \ldots, q_{k-1}\right) = \left(r_1, \ldots, r_{k-1}\right)$ then $q_k \not < r_{k}$. 
Since $\mathbb Q$ is linearly ordered, $p_1 \leqslant r_1$ and, for all $j \in \left\{2, \ldots, n+1 -i\right\}$, if $\left(r_1, \ldots, r_{j-1}\right) = \left(p_1, \ldots, p_{j-1}\right)$ then $p_j \leqslant r_j$. Likewise, $r_1 \leqslant q_1$ and, for all $k \in \left\{2, \ldots, n+1 -i\right\}$, if $\left(q_1, \ldots, q_{k-1}\right) = \left(r_1, \ldots, r_{k-1}\right)$ then $r_k \leqslant q_{k}$. If $p_1 < r_1$, then we have $p_1 < r_1 \leqslant q_1$, and so $(p, q) \in L_1$, which means $\left(p^b, q^d\right) \in U_1 \subseteq R_{-1} \subseteq R_i$. If $r_1 < q_1$, then $p_1 \leqslant r_1 < q_1$, so again $\left(p^b, q^d\right) \in U_1 \subseteq R_{-1}\subseteq R_i$. Finally, suppose $p_1 = r_1 = q_1$. There are four cases:
\begin{enumerate}[(a)]
\item  There exist $\ell, m \in \left\{2, \ldots, n+1 -i\right\}$ such that $\left(p_{1}, \ldots, p_{\ell-1}\right) = \left(r_1, \ldots, r_{\ell-1}\right)$, $p_{\ell} < r_{\ell}$, $\left(q_1, \ldots, q_{m - 1}\right) = \left(r_1, \ldots, r_{m - 1}\right)$ and $r_{m} < q_{m}$; that is, $\left(p, r\right) \in L_\ell$ and $\left(r, q\right) \in L_m$.  Therefore, by Lemma~\ref{lem:transitivity_Tis}, we get
$\left(p, q\right) \in L_{\textnormal{min}\{\ell, m\}}$, and thus $\left(p^b, q^d\right) \in U_{\textnormal{min}\{\ell, m\}} \subseteq R_{-i}\subseteq R_i$.
\item  There exists $\ell \in \left\{2, \ldots, n+1 -i\right\}$ such that $\left(p_{1}, \ldots, p_{\ell-1}\right) = \left(r_1, \ldots, r_{\ell-1}\right)$ and $p_{\ell} < r_{\ell}$, and for all $k\in \left\{2, \ldots, n+1 -i\right\}$, we have $r_k = q_k$. In this case, $\left(p_{1}, \ldots, p_{\ell-1}\right) = \left(q_1, \ldots, q_{\ell-1}\right)$ and $p_{\ell} < r_{\ell} = q_\ell$. Hence, $(p, q) \in L_\ell$, which gives $\left(p^b, q^d\right) \in U_\ell \subseteq R_{-i}\subseteq R_i$. 
\item  There exists $m \in \left\{2, \ldots, n+1 -i\right\}$ such that $\left(q_{1}, \ldots, q_{m-1}\right) = \left(r_1, \ldots, r_{m-1}\right)$ and $r_{m} < q_{m}$, and for all $j\in \left\{1, \ldots, n+1 -i\right\}$, we have $p_j = r_j$. Then $\left(p_{1}, \ldots, p_{m-1}\right) = \left(q_1, \ldots, q_{m-1}\right)$ and $p_{m} = r_{m} < q_m$. Hence, $(p, q) \in L_m$, and so $\left(p^b, q^d\right) \in U_m\subseteq R_{-i} \subseteq R_i$. 
\item  For all $j, k \in \left\{1, \ldots, n+1 -i\right\}$, we have $p_j = r_j$ and $q_k = r_k$. This implies $(q, p) \notin L_j$ for all $j \in \left\{1, \ldots, n+1 -i\right\}$, so $\left(q^{-d}, p^b\right) \notin U_j$ for all $j \in \left\{1, \ldots, n+1 -i\right\}$. Hence, it follows that $\left(q^{-d}, p^b\right) \notin R_{-i}$, which means $\left(p^b, q^{-d}\right) \in \left(R_{-i}\right)^{c\smile}$, i.e.,  $\left(p^b, \alpha\left(q^{d}\right)\right) \in \left(R_{-i}\right)^{c\smile}$. We thus get $\left(p^b, q^d\right) \in \left(R_{-i}\right)^{c\smile} \circ \alpha = {\sim}R_{-i} = R_i$. 
\end{enumerate}
\end{proof}

Given $n \geqslant 1$, in order for the relations $\{\,R_i \mid -(n+1)  \leqslant i \leqslant n+1 \,\}$ to give a representation of the Sugihara chain $\mathbf{S}_{2n+3}$, we need to show the following:
$$
R_i \circ R_j = \begin{cases} 
R_i & \text{if }\, |j|<|i| \\ 
R_j & \text{if }\, |i|< |j| \\ 
R_{\text{min}\{i,j\}} & \text{if }\, |j|=|i|
\end{cases}
$$
We will use a series of lemmas to show that the above holds. Each lemma will involve a number of cases. 
We will frequently make use of the fact that composition is order-preserving in both coordinates, and that in all cases the relations are greater than or less than the monoid identity $\leqslant_X$. 

\begin{lemma}\label{lem:absi>absj}
If $|i|>|j|$ then $R_i \circ R_j = R_i$. 
\end{lemma} 
\begin{proof}
If $i = -n - 1$ (i.e. $R_i =\varnothing$) then we clearly have $R_i \circ R_j = R_i$. The case for $j=0$ follows from Lemma \ref{lem:important_eq_injective_map}. So we consider the following cases:
\\

\noindent
\underline{Case 1:} $-n - 1< i<j<0$. Since $R_j \subseteq {\leqslant_X}$, we get $R_i \circ R_j \subseteq R_i$. Let $\left(p^b,q^d\right) \in R_i$. Hence, there exists $k \in \left\{1, \ldots,  n+1+i\right\}$ such that $\left(p^b,q^d\right) \in U_k$, i.e., there exists $k \in \left\{1, \ldots,  n+1+i\right\}$ such that $\left(p, q\right) \in L_k$. There are two cases:

\begin{enumerate}[(a)]

\item $p_1 < q_1$. By the density of $\mathbb Q$, there is some $r_1 \in \mathbb Q$ such that $p_1 < r_1 < q_1$. Let $r = (r_1, 0,\ldots, 0)$. Then $\left(p, r\right) \in L_1$ and $(r, q) \in L_1$, i.e., $\left(p^b, r^b\right) \in U_1\subseteq R_i$ and $\left(r^b, q^d\right) \in U_1\subseteq R_i\subseteq R_j$. We thus have $\left(p^b, q^d\right) \in R_i\circ R_j$. 

\item There is some $k \in \left\{2, \ldots, n+1+i\right\}$ such that $(p_1,\dots,p_{k-1})=(q_1,\dots,q_{k-1})$ and $p_k < q_k$. By the density of  $\mathbb{Q}$, there exists $r_k$ such that $p_k < r_k < q_k$. Let $r=(p_1, \dots, p_{k-1}, r_k, 0, \dots, 0)$. Clearly $(p^b,r^d) \in U_k \subseteq R_i$. Now notice that $r=(q_1,\dots,q_{k-1},r_k,0,\dots,0)$, and therefore $(r^d,q^d) \in U_k \subseteq R_i \subseteq R_j$. This gives us $(p^b,q^d) \in R_i \circ R_j$ and so $R_i \circ R_j = R_i$.
\end{enumerate}
%\\ 

\noindent
\underline{Case 2:} $-n-1 < i<0$, $j>0$. Since ${\leqslant_X} \subseteq R_j$, we get $R_i \subseteq R_i \circ R_j$. Let $\left(p^b,q^d\right) \in R_i \circ R_j$. Then there is some $r^e$ such that $\left(p^b, r^e\right) \in R_i$ and $\left(r^e, q^d\right) \in R_j$. From the first part we get $\left(p^b, r^e\right) \in U_k$ for some $k \in \left\{1, \ldots, n+1 +i\right\}$, which means $(p, r) \in L_k$ for some  $k \in \left\{1, \ldots, n+1 +i\right\}$. The second part gives $\left(r^e, q^d\right) \in {\sim}R_{-j} = \left(R_{-j}\right)^{c\smile}\circ\alpha$, and so $\left(q^{-d}, r^b\right) \notin R_{-j}$. This implies that $\left(q^{-d}, r^e\right) \notin U_\ell$ for all $\ell \in \left\{1, \ldots, n+1 -j\right\}$. That is, $q_1 \not < r_1$ and, for all $\ell \in \left\{2, \ldots, n+1 -j\right\}$, if $\left( q_1, \ldots, q_{\ell-1}\right) = \left(r_1, \ldots, r_{\ell -1}\right)$ then $q_\ell \not < r_\ell$. Since $\mathbb Q$ is linearly ordered, we get $r_1 \leqslant q_1$ and, for all $\ell \in \left\{2, \ldots, n+1 -j\right\}$, if $\left( q_1, \ldots, q_{\ell-1}\right) = \left(r_1, \ldots, r_{\ell -1}\right)$ then $r_\ell \leqslant q_1$. If $p_1 < r_1$, then $p_1 < r_1 \leqslant q_1$, and so $\left(p, q\right) \in L_1$, i.e., $\left(p^b, q^d\right) \in U_1 \subseteq R_i$. If $r_1 < q_1$, we have $p_1 \leqslant r_1 < q_1$, so again $\left(p, q\right) \in L_1$, i.e., $\left(p^b, q^d\right) \in U_1 \subseteq R_i$. If $p_1 = r_1 = q_1$, then it must be the case that $\left(p_1, \ldots, p_{m-1}\right) = \left(r_1, \ldots, r_{m-1}\right) = \left(q_1, \ldots, q_{m-1}\right)$ and $p_m< q_m$ for some $m \in \left\{2, \ldots, k\right\}$. Hence, $(p, q) \in L_m$, and so $\left(p^b, q^d\right) \in U_m \subseteq R_i$. 
\\

\noindent
\underline{Case 3:} $i>0$, $-n-1 < j<0$. Since $R_j \subseteq {\leqslant_X}$ we get $R_i\circ R_j \subseteq R_i$. Let $\left(p^b,q^d\right) \in R_i$. Then $\left(p^b,q^d\right) \in {\sim}R_{-i} = \left(R_{-i}\right)^{c\smile}\circ \alpha$, and so $\left(q^{-d}, p^b\right) \notin R_{-i}$. This means that $\left(q^{-d}, p^b\right) \notin U_k$ for all $k \in \left\{1, \ldots, n+1 -i\right\}$. That is, $q \not < p_1$ and, for all $k \in \left\{2, \ldots, n+1 -i\right\}$, if $\left(p_1, \ldots, p_{k-1}\right) = \left(q_1,\ldots, q_{k-1}\right)$ then $q_k \not < p_k$. We therefore obtain $p_1 \leqslant q_1$ and, for all $k \in \left\{2, \ldots, n+1 -i\right\}$, if $\left(p_1, \ldots, p_{k-1}\right) = \left(q_1,\ldots, q_{k-1}\right)$ then $p_k \leqslant q_k$. Suppose $p_1 < q_1$. By the density of  $\mathbb Q$, there is some $r_1 \in \mathbb Q$ such that $p_1 < r_1 < q_1$. Now let $r = \left(p_1, 0, \ldots, 0\right)$. Then $(p, r) \in L_1$ and $(r, q) \in L_1$, i.e., $\left(p^b, r^b\right) \in U_1 \subseteq R_{-i} \subseteq R_i$ and $\left(r^b, q^d\right) \in U_1 \subseteq R_{-i} \subseteq R_j$. Hence, $\left(p^b, q^d\right) \in R_i \circ R_j$. On the other hand, suppose $p_1 = q_1$. We have two cases: 

\begin{enumerate}[(a)]
\item There is some $\ell \in \left\{2, \ldots, n+1 -i\right\}$ such that $\left(p_1, \ldots, p_{\ell - 1}\right) = \left(q_1, \ldots, q_{\ell - 1}\right)$ and $p_\ell < q_\ell$. Then  $\left(p, q\right) \in L_{\ell}$, and so $\left(p^b, q^d\right) \in U_\ell \subseteq R_{-i} \subseteq R_j$. Since $\leqslant_X\, \subseteq R_i$, we have $\left(p^b, p^b\right) \in R_i$, and therefore we obtain $\left(p^b, q^d\right) \in R_i \circ R_j$.
\item  For all $k \in \left\{2, \ldots, n+1 -i\right\}$, we have $p_k = q_k$. Let $$r = \left(q_1, \ldots, q_{n+1-i}, q_{n+2-i}-1,\ldots, q_n-1\right).$$ Then $(r, q) \in L_{n+2 -i}$, and so $\left(r^b, q^d\right) \in U_{n+2-i} \subseteq \bigcup_{k=1}^{n+1+j} \subseteq R_j$. Note now that we also have $r = \left(p_1, \ldots, p_{n+1-i}, q_{n+2-i}-1,\ldots, q_n-1\right)$, so $\left(r^{-b}, p^b\right) \notin U_{k}$ for all $k \in \left\{1, \ldots, n+1 -i\right\}$. Hence, $\left(r^{-b}, p^b\right) \notin R_{-i}$, which implies that $\left(p^b, r^{-b}\right) \in \left(R_{-i}\right)^{c\smile}$. It therefore follows that $\left(p^b, r^{b}\right) \in \left(R_{-i}\right)^{c\smile} \circ \alpha = {\sim}R_{-i} = R_i$. From this we get $\left(p^b, q^d\right) \in R_i \circ R_j$. 
\end{enumerate}

\noindent
\underline{Case 4:} $i> j > 0$. Since ${\leqslant_X} \subseteq R_j$ we get $R_i \subseteq R_i \circ R_j$. Let $\left(p^b,q^d\right) \in R_i \circ R_j$. Then there is some $r^e$ such that $\left(p^b, r^e\right) \in R_i$ and $\left(r^e, q^d\right) \in R_j$. But $R_j\subseteq R_i$ and $R_i$ is transitive by Lemma \ref{lem:Ri_transitive}, so $\left(p^b, q^d\right) \in R_i$.   
\end{proof}

\begin{lemma}
If $|j|>|i|$ then $R_i \circ R_j = R_j$. 
\end{lemma}
\begin{proof}
If $j = -n -1$, then $R_j = \varnothing$, and so $R_i \circ R_j = R_j$. If $i = 0$, we have $R_i = \;\leq_X$, so $R_i \circ R_j = R_j$ by Lemma \ref{lem:order_identity}. We now consider the following cases:
\\

\noindent
\underline{Case 1:} $-n - 1 < j < i < 0$. Since $R_i \subseteq {\leqslant_X}$, we get $R_i \circ R_j \subseteq R_j$. Let $\left(p^b,q^d\right) \in R_j$. Then $\left(p, q\right) \in L_k$ for some  $k \in \left\{1, 2, \ldots, n+1+ j\right\}$. That is, $p_1 < q_1$ or $(p_1,\dots,p_{k-1})=(q_1,\dots,q_{k-1})$ and $p_k < q_k$ for some  $k \in \left\{2, \ldots, n+1+j\right\}$. If $p_1 < q_1$, we use the density of $\mathbb Q$ to find some $r_1 \in \mathbb Q$ such that $p_1 < r_1 < q_1$. Then we set $r = \left(r_1, 0 \ldots, 0\right)$ and obtain $\left(p^b, r^b\right) \in U_1 \subseteq R_i$ and $\left(r^b, q^d\right) \in U_1 \subseteq R_{j}$. Hence, $\left(p^b, q^d\right) \in R_i \circ R_j$.  Similarly, if $(p_1,\dots,p_{k-1})=(q_1,\dots,q_{k-1})$ and $p_k < q_k$ for some  $k \in \left\{2, \ldots, n+1+j\right\}$, we can find some $r_k \in \mathbb Q$ such that $p_k < r_k < q_k$. Setting $r = \left(q_1, \ldots, q_{k-1}, r_k, 0, \ldots, 0\right)$, we obtain $\left(p^b, r^b\right) \in U_k \subseteq R_j\subseteq R_i$ and $\left(r^b, q^d\right) \in U_k \subseteq R_{j}$, which means $\left(p^b, q^d\right) \in R_i \circ R_j$. \\

\noindent
\underline{Case 2:} $i > 0, -n - 1< j < 0$. Since ${\leqslant_X} \subseteq R_i$, we get $R_j \subseteq R_i \circ R_j$. Let $\left(p^b,q^d\right) \in R_i \circ R_j$. Then there is some $r^e$ such that $\left(p^b, r^e\right) \in R_i$ and $\left(r^e, q^d\right) \in R_j$. The latter gives 
$(r, q) \in L_k$ for some  $k \in \left\{1, \ldots, n+1 +j\right\}$. The former gives $\left(p^b, r^e\right) \in {\sim}R_{-i} = \left(R_{-i}\right)^{c\smile}\circ\alpha$, and so $\left(r^{-e}, p^b\right) \notin R_{-i}$. Hence, $\left(r^{-e}, p^b\right) \notin U_\ell$ for all $\ell \in \left\{1, \ldots, n+1 -i\right\}$. That is, $r_1 \not < p_1$ and, for all $\ell \in \left\{2, \ldots, n+1 -i\right\}$, if $\left(p_1, \ldots, p_{\ell-1}\right) = \left(r_1, \ldots, r_{\ell -1}\right)$ then $r_\ell \not < p_\ell$. Since $\mathbb Q$ is linearly ordered, we get $p_1 \leqslant r_1$ and, for all $\ell \in \left\{2, \ldots, n+1 -j\right\}$, if $\left( p_1, \ldots, p_{\ell-1}\right) = \left(r_1, \ldots, r_{\ell -1}\right)$ then $p_\ell \leqslant r_1$. If $r_1 < q_1$, then $p_1 \leqslant r_1 < q_1$, and so  $\left(p^b, q^d\right) \in U_1 \subseteq R_j$. If $p_1 < r_1$, we get $p_1 < r_1 \leqslant q_1$, so again $\left(p^b, q^d\right) \in U_1 \subseteq R_j$. If $p_1 = r_1 = q_1$, then $\left(p_1, \ldots, p_{m-1}\right) = \left(r_1, \ldots, r_{m-1}\right) = \left(q_1, \ldots, q_{m-1}\right)$ and $p_m< q_m$ for some $m \in \left\{2, \ldots, k\right\}$. Hence, $(p, q) \in L_m$, and so $\left(p^b, q^d\right) \in U_m \subseteq R_j$. 
\\

\noindent
\underline{Case 3:} $i < 0, j > 0$. Since $R_i \subseteq {\leqslant_X}$, we have $R_i \circ R_j \subseteq R_j$. Let $\left(p^b,q^d\right) \in R_j$. Then $\left(p^b,q^d\right) \in {\sim}R_{-j} = \left(R_{-j}\right)^{c\smile}\circ \alpha$, and so $\left(q^{-d}, p^b\right) \notin R_{-j}$. This means that $\left(q, p\right) \notin L_k$ for all $k \in \left\{1, \ldots, n+1 -j\right\}$. Hence, $q_1 \not < p_1$ and, for all $k \in \left\{2, \ldots, n+1 -j\right\}$, if $\left(q_1, \ldots, q_{k-1}\right) = \left(p_1, \ldots, p_{k-1}\right)$ then $q_k \not < p_k$. Since $\mathbb Q$ is linearly ordered, $p_1 \leqslant q_1$ and, for all  $k \in \left\{2, \ldots, n+1 -j\right\}$, if $\left(q_1, \ldots, q_{k-1}\right) = \left(p_1, \ldots, p_{k-1}\right)$ then $p_{k} \leqslant q_k$. Suppose $p_1 < q_1$. By the density of $\mathbb Q$, there is some $r_1\in \mathbb Q$ such that $p_1 < r_1 < q_1$. Let $r = \left(r_1, 0, \ldots, 0\right)$. Then we have $\left(p^b, r^b\right) \in U_1 \subseteq R_{-j} \subseteq R_{i}$ and $\left(r^b, q^d\right) \in U_1 \subseteq R_{-j} \subseteq R_{j}$. Hence, $\left(p^b, q^d\right) \in R_i \circ R_j$. Now suppose $p_1 = q_1$. Then there are two cases: 

\begin{enumerate}[(a)]
\item There is some $\ell \in \left\{2, \ldots, n+ 1-j\right\}$ such that $\left(p_1, \ldots, p_{\ell-1}\right) = \left(q_1, \ldots, q_{\ell-1}\right)$ and $p_\ell < q_\ell$. Then we have $\left(p, q\right) \in L_{\ell}$, and so $\left(p^b, q^d\right) \in U_\ell \subseteq R_{-j} \subseteq R_i$. But $\left(q^d, q^d\right) \in R_j$ since $\leqslant_X\, \subseteq R_j$, so it follows that $\left(p^b, q^d\right) \in R_i \circ R_j$. 

\item For all $k \in \left\{1,\ldots, n+1 -j\right\}$, we have $p_k = q_k$. Let $$r = \left(p_1, \ldots, p_{n+1-j}, p_{n+2-j}+1, \ldots, p_n + 1\right).$$ Then $\left(p, r\right) \in L_{n+2 -j}$, and so $\left(p^b, r^b\right) \in U_{n+ 2-j} \subseteq \bigcup_{k=1}^{n+1+i}U_k = R_i$. We also have $\left(q, r\right) \notin L_{k}$ for all $k \in \left\{1,\ldots, n+1 -j\right\}$.  Hence, $\left(q^{-d}, r^b\right) \notin U_k$ for all $k \in \left\{1, \ldots, n+1 -j\right\}$, which means $\left(q^{-d}, p^b\right) \notin R_{-j}$. We thus have $\left(r^b, q^{-d}\right) \in \left(R_{-j}\right)^{c\smile}$,
and thus, $\left(r^b, q^d\right) \in \left(R_{-j}\right)^{c\smile}\circ \alpha = {\sim}R_{-j} = R_j$. This gives us $\left(p^b, q^d\right) \in R_i \circ R_j$. 
 \end{enumerate}

\noindent
\underline{Case 4:} $j > i > 0$. Since ${\leqslant_X} \subseteq R_i$ we have $R_j \subseteq R_i \circ R_j$. Let $\left(p^b,q^d\right) \in R_i \circ R_j$. Then there is some $r^e$ such that $\left(p^b, r^e\right) \in R_i$ and $\left(r^e, q^d\right) \in R_j$. But $R_i\subseteq R_j$ and $R_j$ is transitive by Lemma \ref{lem:Ri_transitive}, so $\left(p^b, q^d\right) \in R_j$.   
\end{proof}

\begin{lemma}
If $|j|=|i|$ then $R_i \circ R_j = R_{\text{min}\{i,j\}}$. 
\end{lemma}
\begin{proof}
The claim clearly follows for the cases $i=j=0$ and $i=j =-n-1$. So consider the following cases: 
\\

\noindent
\underline{Case 1:} $-n-1 < i = j < 0$. Let $\left(p^b, q^d\right) \in R_i\circ R_j = R_i \circ R_i$. Then there exists $r^b$ such that $\left(p^b, r^e\right) \in R_i$ and $\left(r^e, q^d\right) \in R_i$. By Lemma \ref{lem:Ri_transitive}, $R_i$ is transitive, so $\left(p^b, q^d\right) \in R_i = R_j = R_{\text{min}\{i, j\}}$. Now let $\left(p^b, q^d\right) \in R_{\text{min}\{i, j\}} = R_i = R_j$. 
Then $(p, q) \in L_k$ for some $k \in \left\{1, \ldots, n+1 + i\right\}$. There are two cases: 

\begin{enumerate}[(a)]
\item  $p_1 < q_1$. By the density of  $\mathbb Q$, there exists $r_1\in \mathbb Q$ such that $p_1 < r_1 < q_1$. Let $r = \left(r_1, 0, \ldots, 0\right)$. Then $\left(p^b, r^b\right) \in U_1 \subseteq R_{i}$ and $\left(r^b, q^d\right) \in U_1 \subseteq R_{j}$. Hence, $\left(p^b, q^d\right) \in R_i \circ R_j$. 
\item  There exists $k \in \left\{2, \ldots, n+1 +i\right\}$ such that $(p_1,\dots,p_{k-1})=(q_1,\dots,q_{k-1})$ and $p_k < q_k$. By the density of  $\mathbb{Q}$, there exists $r_k$ such that $p_k < r_k < q_k$. Set $r=(p_1, \dots, p_{k-1}, r_k, 0, \dots, 0)$. Then $(p^b,r^d) \in U_k \subseteq R_i$ and $(r^d,q^d) \in U_k \subseteq R_j$. This gives us $(p^b,q^d) \in R_i \circ R_j$. 
\end{enumerate}

\noindent
\underline{Case 2:} $-n-1 < i< 0$, $j >0$. Since $\leqslant_X \;\subseteq R_j$, we have $R_{\text{min}\{i, j\}} = R_i \subseteq R_i\circ R_j$. Now let $(p^b, q^d) \in R_i \circ R_j$. Then there is some $r^b$ such that $\left(p^b, r^e\right) \in R_i$ and $\left(r^e, q^d\right) \in R_j$.
The former gives $\left(p, r\right) \in L_k$ for some $k \in \left\{1, \ldots, n+1 + i\right\}$. The latter gives $\left(q^{-d}, r^e\right) \notin R_{-j}$, and so $(q, r) \notin L_\ell$ for all $\ell \in \left\{1, \ldots, n+1 -j\right\}$. Hence, $q_1 \not < r_1$ and, for all $\ell \in \left\{2, \ldots, n+1 -j\right\}$, if $\left(r_1, \ldots, r_{\ell-1}\right) = \left(q_1, \ldots, q_{\ell-1}\right)$ then $q_\ell\not < r_\ell$. Since $\mathbb Q$ is linearly ordered, we obtain $r_1 \leqslant q_1$ and, for all $\ell \in \left\{2, \ldots, n+1 -j\right\}$, if $\left(r_1, \ldots, r_{\ell-1}\right) = \left(q_1, \ldots, q_{\ell-1}\right)$ then $r_\ell\leqslant q_\ell$. If $p_1 < r_1$, then $p_1 < r_1 \leqslant q_1$, so $\left(p^b, q^d\right) \in U_1 \subseteq R_i = R_{\text{min}\{i, j\}}$. If $r_1 < q_1$, we get $p_1 \leqslant r_1 < q_1$, which implies $\left(p^b, q^d\right) \in U_1 \subseteq R_i = R_{\text{min}\{i, j\}}$. Finally, if $p_1 = r_1 = q_1$, then $\left(p_1, \ldots, p_{m-1}\right) = \left(r_1, \ldots, r_{m-1}\right) = \left(q_1, \ldots, q_{m-1}\right)$ and $p_m < q_m$ for some $m \in \left\{2, \ldots, k\right\}$. Hence, $(p, q) \in L_m$, and so $\left(p^b, q^d\right) \in U_m \subseteq R_i = R_{\text{min}\{i, j\}}$. 
\\

\noindent
\underline{Case 3:} $-n-1 < j < 0, i> 0$. Since $\leqslant_X \;\subseteq R_i$, we have $R_{\text{min}\{i, j\}} = R_j \subseteq R_i\circ R_j$. Let $(p^b, q^d) \in R_i \circ R_j$. Then there exists some $r^b$ such that $\left(p^b, r^e\right) \in R_i$ and $\left(r^e, q^d\right) \in R_j$. The second part gives $\left(r, q\right) \in L_\ell$ for some $\ell \in \left\{1, \ldots, n+1 + j\right\}$. The first part gives $\left(r^{-e}, p^b\right) \notin R_{-i}$, which means $(r, p) \notin L_k$ for all $k \in \left\{1, \ldots, n+1 -i\right\}$. Hence, $r_1 \not < p_1$ and, for all $k \in \left\{2, \ldots, n+1 -i\right\}$, if $\left(r_1, \ldots, r_{k-1}\right) = \left(p_1, \ldots, p_{k-1}\right)$ then $r_k\not < p_k$. Since $\mathbb Q$ is linearly ordered, we get $p_1 \leqslant r_1$ and, for all $k \in \left\{2, \ldots, n+1 -i\right\}$, if $\left(r_1, \ldots, r_{k-1}\right) = \left(p_1, \ldots, p_{k-1}\right)$ then $p_k\leqslant r_k$. If $r_1 < q_1$, then $p_1 \leqslant r_1 < q_1$, and so $\left(p^b, q^d\right) \in U_1 \subseteq R_j = R_{\text{min}\{i, j\}}$. If $p_1 < r_1$, then it follows that $p_1 < r_1 \leqslant q_1$, which gives $\left(p^b, q^d\right) \in U_1 \subseteq R_j = R_{\text{min}\{i, j\}}$. Finally, if $p_1 = r_1 = q_1$, then $\left(p_1, \ldots, p_{m-1}\right) = \left(r_1, \ldots, r_{m-1}\right) = \left(q_1, \ldots, q_{m-1}\right)$ and $p_m < q_m$ for some $m \in \left\{2, \ldots, \ell\right\}$. Hence, $(p, q) \in L_m$, and so $\left(p^b, q^d\right) \in U_m \subseteq R_j = R_{\text{min}\{i, j\}}$. 
\\
\noindent
\underline{Case 4:} $i=j > 0$. The left-to-right inclusion follows from the transitivity of $R_i = R_j$. Now let $\left(p^b, q^d\right) \in R_{\text{min}\left\{i, j\right\}} = R_i = R_j$.  Then $\left(p^b, p^b\right) \in\; \leqslant_X\, \subseteq R_i$, and hence $\left(p^b, q^d\right) \in R_i\circ R_j$.  
\end{proof}

Combining the three lemmas above gives us the following result.  
\begin{proposition}\label{prop:comp-for-Ri}
	Let $n \geqslant 1$ and $i, j \in \{-(n+1),\dots,n+1\}$. Then 
	$$
	R_i \circ R_j = \begin{cases} 
		R_i & \text{if }\, |j|<|i| \\ 
		R_j & \text{if }\, |i|< |j| \\ 
		R_{\text{min}\{i,j\}} & \text{if }\, |j|=|i|
	\end{cases}
	$$
	
\end{proposition}

Let $n \geqslant 1$ and consider $\mathcal{R}_{2n+3}= \{\, R_i \mid -(n+1) \leqslant i \leqslant n+1\,\}$. Then define the algebra 
$$\mathbb{S}_{2n+3}:=\langle \mathcal{R}_{2n+3}, \cap, \cup, \circ, \Rightarrow, R_0, {\sim}\rangle$$ 
where $\Rightarrow$ is defined, as in the abstract case, in terms of $\cap$, $\cup$ and ${\sim}$. 
\begin{theorem}\label{thm:main}
Let $n \geqslant 1$. Then the algebra 
$\mathbb{S}_{2n+3}$ is an odd Sugihara chain and $\mathbf{S}_{2n+3} \cong \mathbb{S}_{2n+3}$. 
\end{theorem}
\begin{proof}
The fact that the ordered set forms a chain is Lemma~\ref{lem:Ri-chain}. Comparing  Proposition~\ref{prop:comp-for-Ri} with the definition of $\mathbf{S}_{2n+3}$ in Section~\ref{sec:sugihara} shows that the monoid operation is defined as it should be for an odd Sugihara chain. The fact that ${\sim}$ is order-reversing follows from the fact that $R_i={\sim}R_{-i}$ and Lemma~\ref{lem:Ri-chain}. The isomorphism is given by $\varphi: S_{2n+3} \to \mathcal{R}_{2n+3}$ where $\varphi(a_i)=R_i$. 
\end{proof}

\section{Representing all Sugihara monoids via weakening relations}\label{sec:RSM}

The main result of this section is to show that all Sugihara monoids can be represented as algebras of binary relations with the monoid operation given by relational composition. Since the algebras of relations are constructed using the method of Section~\ref{subsec:construct}, 
the relations will be weakening relations. We begin with an important definition. 

\begin{definition}\label{def:RSM}
A Sugihara monoid $\mathbf{A}$ is \emph{representable} if it is isomorphic to the direct reduct of an element of $\mathsf{RDInFL}$.    
\end{definition}
The class of representable Sugihara monoids will be denoted by $\mathsf{RSM}$. We note that since $\mathsf{RDInFL}$ is closed under $\mathbb{I}$, $\mathbb{S}$ and $\mathbb{P}$ (see Definition~\ref{def:RDInFL}), we will also have that $\mathsf{RSM}=\mathbb{ISP}(\mathsf{RSM})$. 

From here, we proceed as follows. 
We show that the class of representable DInFL-algebras is closed under ultraproducts.   Then, using the closure of $\mathsf{RDInFL}$ under ultraproducts, we are able to show that the Sugihara monoids $\mathbf{S}$ and $\mathbf{S}^*$ are representable.
Since those two algebras generate the quasivariety $\mathsf{SM}$, we will then be able to  conclude that $\mathsf{SM} \subseteq \mathsf{RSM}$.

\subsection{Closure of $\mathsf{RDInFL}$ under ultraproducts}\label{sec:ultraproducts} 

Here we show that the class of representable DInFL-algebras
(see Definition~\ref{def:RDInFL})
is closed under taking ultraproducts. The ultraproduct was first introduced by \L o\'s~\cite{Los55} in 1955. 
The construction is now very well-known, both for algebras and other first-order structures (cf.~\cite[IV-6 and V-2]{BS2012}). 

Let $I$ be a set and let $\left\{A_i \mid i \in I\right\}$ be a family of non-empty sets indexed by $I$. Recall that if $a, b \in \prod\{A_i\mid i \in I\}$, then the subset
\[
\llbracket a=b\rrbracket = \left\{i \in I \mid a(i) = b(i)\right\}
\] 
of $I$ is called the \emph{equaliser} of $a$ and $b$. 

If $\mathcal{F}$ is a filter on $I$, then we may define a binary relation $\theta_\mathcal{F}$ on the Cartesian product $\prod\left\{A_i \mid i \in I\right\}$ by setting, for all $a, b \in \prod\left\{A_i \mid i \in I\right\}$, 
\[
\left(a, b\right) \in \theta_{\mathcal{F}} \qquad \textnormal{ iff } \qquad \llbracket a=b\rrbracket \in \mathcal{F}.
\]
It is straightforward to check that $\theta_\mathcal{F}$ is an equivalence relation on $\prod\left\{A_i \mid i \in I\right\}$. 

\begin{lemma}\label{lem:theta_congruence}
Let $I$ be a set and let $\mathbf{A}_i$ be an algebra, for each $i \in I$. For each filter $\mathcal{F}$ on $I$, the relation $\theta_\mathcal{F}$ is a congruence on the direct product $\prod\left\{\mathbf{A}_i \mid i \in I\right\}$.
\end{lemma}

Let $I$ be a set and $\mathbf{A}_i$ an algebra, for each $i \in I$. Recall that if the filter $\mathcal{F}$ in Lemma~\ref{lem:theta_congruence} is an ultrafilter on $I$, then the algebra $\prod\left\{\mathbf{A}_i \mid i \in I\right\}/\theta_\mathcal{F}$ is called the \emph{ultraproduct} of $\left\{\mathbf{A}_i \mid i \in I\right\}$. For a class of algebras $\mathsf{K}$, we denote by $\mathbb{P_U}(\mathsf{K})$ the class of all ultraproducts of members of $\mathsf{K}$. The following theorem is immediate. 

\begin{theorem}\label{thm:ultra_prod_DInFl}
An ultraproduct of a family of DInFL-algebras is again a DInFL-algebra. 
\end{theorem}

Below we outline how to construct the ultraproduct of a set of partially ordered sets with additional structure, each of which can be used to represent a DInFL-algebra. 
Consider an index set $I$ and let $\mathcal{F}$ be an ultrafilter on $I$. For each index $i \in I$, let $\mathbf X_i = \left\langle X_i, \leqslant_i\right\rangle$ be a poset and $E_i$ an equivalence relation on $X_i$ such that $X_i \neq \varnothing$ and ${\leqslant_i} \subseteq E_i$. Consider the Cartesian product $X = \prod_{i \in I}X_i$ and the equaliser $\llbracket x=y\rrbracket = \left\{i \in I \mid x(i) = y(i)\right\}$ of $x, y \in X$. Define a binary relation on $X$ by setting, for all $x, y \in X$,
\[
\left(x, y\right) \in E_\mathcal{F} \qquad \textnormal{iff} \qquad \llbracket x=y\rrbracket\in \mathcal{F}.
\] 
Then $E_\mathcal{F}$ is an equivalence relation on $X$. Now consider the 
set of equivalence classes
$Y=X/E_{\mathcal{F}}$. Define binary relations $\leqslant_Y$ and $E_Y$ on $Y$ by setting, for all $[x], [y] \in Y$,
\begin{align*}
[x] \leqslant_Y [y]  \qquad  & \textnormal{iff} \qquad \left\{i \in I \mid x(i) \leqslant_i y(i)\right\} \in \mathcal{F},\\
\left([x], [y]\right) \in E_Y \qquad & \textnormal{iff} \qquad \left\{i \in I \mid \left(x(i), y(i)\right) \in E_i\right\} \in \mathcal{F}. 
\end{align*}

Suppose further that for each $i \in I$, the map $\alpha_i : X_i \to X_i$ is an order automorphism of $\mathbf X_i$ such that $\alpha_i \subseteq E_i$. Define $\alpha_Y: Y \to Y$ by setting, for each $[x] \in Y$, 
\[
\alpha_Y\left([x]\right) = \left[\alpha_X(x)\right]
\]
where $\left(\alpha_X(x)\right)(i) = \alpha_i\left(x(i)\right)$ for all $x \in X$ and $i \in I$. Since each of the $\alpha_i$ is bijective, we get the following result. 

\begin{lemma}\label{lem:alphaY_inverse}
The map $\alpha_Y: Y \to Y$ has an inverse. Moreover, $\alpha^{-1}_Y\left([x]\right) = \left[\alpha^{-1}_X(x)\right]$ where $\left(\alpha^{-1}_X(x)\right)(i) = \alpha_i^{-1}\left(x(i)\right)$ for all $x \in X$ and $i \in I$. 
\end{lemma}

The above definitions will be used in the results that follow. 

\begin{theorem}\label{thm:ultraproductDInFL}
Let $I$ be an index set and $\mathcal{F}$  an ultrafilter on $I$. For each index $i \in I$, let $\mathbf X_i = \left\langle X_i, \leqslant_i\right\rangle$ be a poset and $E_i$ an equivalence relation on $X_i$ such that $X_i \neq \varnothing$ and ${\leqslant_i} \subseteq E_i$. Suppose further that for each $i \in I$, the map $\alpha_i : X_i \to X_i$ is an order automorphism of $\mathbf X_i$ such that $\alpha_i \subseteq E_i$. Then:
\begin{enumerate}[\normalfont (i)]
\item $\mathbf{Y} = \left\langle Y, \leqslant_Y\right\rangle$ is a poset.
\item $E_Y$ is an equivalence relation on $Y$ such that ${\leqslant_Y} \subseteq E_Y$.
\item $\alpha_Y$ is an order automorphism of $\mathbf{Y}$ such that $\alpha_Y \subseteq E_Y$. 
\item If $0 = \alpha_Y \circ \left(\leqslant_Y\right)^{c\smile} = \left(\leqslant_Y\right)^{c\smile} \circ \alpha_Y$, then $\left\langle \mathsf{Up}\left(\left\langle E_Y, \preceq_Y\right\rangle\right), \cap, \cup, \circ, \leqslant_Y, 0, \sim, -\right\rangle$ is a distributive InFL-algebra. 
\end{enumerate}
\end{theorem}

\begin{proof}
Since all the conditions in items (i) to (iii) can be translated into first-order sentences in the relevant language, they follow from \L o\'s's Theorem. 
Item (iv) follows from part (i) of Theorem~\ref{thm:InFL-algebra}. 
\end{proof}

\begin{remark}
The fact that the structures with which the DInFL-algebras are represented (i.e. posets with an equivalence relation and an order automorphism) are first-order structures, greatly simplifies the proof that $\mathsf{RDInFL}$ is closed under ultraproducts. Theorem~4.5 above gives us the candidate structure from which we can build the algebra in which we will embed the ultraproduct of representable algebras. 
\end{remark}
Before proving that the class $\mathsf{RDInFL}$ is closed under ultraproducts, we need the lemmas below to ensure that we have a candidate for the injective homomorphism needed in Theorem~\ref{thm:Pu(RDInFL)=RDInFL}.  For Lemmas~\ref{lem:Ra-upset} and~\ref{lem:varphi-inj},  and Theorem~\ref{thm:Pu(RDInFL)=RDInFL}, we will be working in the following setting. Let $I$ be an index set and $\mathcal{F}$ an ultrafilter on $I$. For each $i \in I$, let $\mathbf{A}_i$ be a representable distributive InFL-algebra. Let $\mathbf{A}$ denote the direct product $\prod\left\{\mathbf{A}_i \mid i \in I\right\}$, and $\mathbf{A}/\theta_{\mathcal{F}}$ the ultraproduct. 

We know that for each $i \in I$, there is a poset $\left\langle X_i, \leqslant_i\right\rangle$, an equivalence relation $E_i$ on $X_i$ such that ${\leqslant_i} \subseteq E_i$, an order automorphism $\alpha_i$ of $\left\langle X_i, \leqslant_i\right\rangle$ such that $\alpha_i \subseteq E_i$, and an embedding $\varphi_i$ from $\mathbf{A}_i$ into 
$\mathfrak{D}\left(\left\langle E_i, \preceq_i\right\rangle\right)$.
We may assume that $X_i \neq \varnothing$ for all $i \in I$ since the direct product of a family of algebras from $\mathsf{EDInFL}$ is isomorphic
to another such product from which all factors of the form $\mathfrak{D}\left(\left\langle \varnothing, \varnothing\right\rangle\right)$ have
been deleted. Now, for each $[a] \in A/\theta_{\mathcal{F}}$, define
\[
\left([x], [y]\right) \in R_{[a]} \qquad \textnormal{iff} \qquad \left\{i \in I \mid \left(x(i), y(i)\right) \in \varphi_i\left(a(i)\right)\right\} \in \mathcal{F}.
\]

\begin{lemma}\label{lem:Ra-upset}
For any $[a] \in A/\theta_{\mathcal{F}}$, the set $R_{[a]}$ is an upset of $\left\langle E_Y, \preceq_Y\right\rangle$. 
\end{lemma}
\begin{proof}
Let $\left([x], [y]\right) \in R_{[a]}$ and assume $\left([x], [y]\right) \preceq_Y \left([u], [v]\right)$. The first statement implies that we have $\left\{i \in I \mid \left(x(i), y(i)\right) \in \varphi_i(a(i))\right\} \in \mathcal{F}$. From the second part we get $[u] \leqslant_Y [x]$ and $[y] \leqslant_Y [v]$, which means $\left\{i \in I \mid u(i) \leqslant_i x(i)\right\} \in \mathcal{F}$ and $\left\{i \in I \mid y(i) \leqslant_i v(i)\right\} \in \mathcal{F}$. Since $\mathcal{F}$ is closed under taking intersections, we have
\[
J= 
\left\{i \in I \mid \left(x(i), y(i)\right) \in \varphi_i(a(i))\right\} \cap  \left\{i \in I \mid u(i) \leqslant_i x(i)\right\} \cap \left\{i \in I \mid y(i) \leqslant_i v(i)\right\} \in \mathcal{F}.
\]
If we can show that  
$J\subseteq \left\{i \in I \mid \left(u(i), v(i)\right) \in \varphi_i(a(i))\right\}$,
we can use the fact that $\mathcal{F}$ is closed under taking supersets to conclude that $\left\{i \in I \mid \left(u(i), v(i)\right) \in \varphi_i(a(i))\right\} \in \mathcal{F}$, which means $\left([u], [v]\right) \in R_{[a]}$. To this end, let $j\in J$.
Then $\left(x(j), y(j)\right) \in \varphi_j(a(j))$, $u(j) \leqslant_j x(j)$ and $y(j) \leqslant_j v(j)$. From the second and third part we obtain $\left(x(j), y(j)\right) \preceq_j \left(u(j), v(j)\right)$. Hence, since $\left(x(j), y(j)\right) \in \varphi_j(a(j))$ and $\varphi_j\left(a(j)\right)$ is an upset of $\left\langle E_j, \preceq_j\right\rangle$, we have $\left(u(j), v(j)\right) \in \varphi_j\left(a(j)\right)$. 
\end{proof}

With the above lemma in hand, we can define a function $\varphi: A/\theta_\mathcal{F} \to\mathsf{Up}\left(\langle E_Y, \preceq_Y\rangle \right)$ by setting
$\varphi\left([a]\right) = R_{[a]}$ for all 
$\left[a\right] \in A/\theta_{\mathcal{F}}$.

\begin{lemma}\label{lem:varphi-inj}
The map $\varphi$ is an injective function from 
$A/\theta_\mathcal{F}$  to $\mathsf{Up}\left(\langle E_Y, \preceq_Y\rangle \right)$.
\end{lemma}
\begin{proof}
By Lemma~\ref{lem:Ra-upset} we know that $\varphi$ is well-defined. Now assume that $[a]\neq[b]$. This implies $(a,b) \notin \theta_\mathcal{F}$ and hence $\{ i \in I \mid a(i)=b(i) \} \notin \mathcal{F}$. We get $\{ i \in I \mid a(i)\neq b(i) \} \in \mathcal{F}$, and since each of the $\varphi_i$ is injective, it follows that $\{ i \in I \mid a(i)\neq b(i) \} \subseteq \{i\in I \mid \varphi_i(a(i)) \neq \varphi_i (b(i))\}$. But  $\mathcal{F}$ is closed under taking supersets, so $\{i\in I \mid \varphi_i(a(i)) \neq \varphi_i (b(i))\} \in \mathcal{F}$. Now 
$$\{i\in I \mid \varphi_i(a(i)) \neq \varphi_i (b(i))\} = \{ i \in I \mid \varphi_i(a(i))\nsubseteq \varphi_i(b(i))\} \cup
\{ i \in I \mid \varphi_i(a(i))\nsupseteq \varphi_i(b(i))\}\in \mathcal{F}.
$$
Hence, either 
$\{i \in I \mid \varphi_i(a(i))\nsubseteq \varphi_i(b(i))\} \in \mathcal{F}$ or $
\{i \in I \mid \varphi_i(a(i))\nsupseteq \varphi_i(b(i))\} \in \mathcal{F}$, since $\mathcal{F}$ is an ultrafilter. Let $V=\{ i \in I \mid \varphi_i(a(i))\nsubseteq \varphi_i(b(i))\}$ and assume $V \in \mathcal{F}$. Since $\mathcal{F}$ is a proper filter, $V \neq \varnothing$. For each $i\in V$, there exists $(s_i,t_i) \in E_i$ such that $(s_i,t_i) \in \varphi_i(a(i))$ and $(s_i,t_i) \notin \varphi_i(b(i))$. Since each $X_i \neq \varnothing$, we have $r_i \in X_i$ for all $i \in I$. Now define two elements of $\prod_{i \in I} X_i$ as follows: 
$$x(i)=\begin{cases} s_i & i \in V  \\ r_i  & i \notin V\end{cases}\qquad \text{ and }\qquad 
y(i)=\begin{cases} t_i & i \in V  \\ r_i  & i \notin V.\end{cases}
$$
It is not difficult to see that 
$$V \subseteq \{i \in I \mid (x(i),y(i)) \in \varphi_i(a(i))\} \quad \text{and} \quad  V \subseteq \{ i \in I \mid (x(i),y(i)) \notin \varphi_i(b(i))\},$$ and hence $\{i \in I \mid (x(i),y(i)) \in \varphi_i(a(i))\}\in \mathcal{F}$ and $\{ i \in I \mid (x(i),y(i)) \notin \varphi_i(b(i))\} \in \mathcal{F}$. 
The second part implies that $\{ i \in I \mid (x(i),y(i)) \in \varphi_i(b(i))\}^c \in \mathcal{F}$, which means $$\{ i \in I \mid (x(i),y(i)) \in \varphi_i(b(i))\} \notin \mathcal{F}.$$
Hence we have $([x],[y])\in \varphi([a])$ but $([x],[y]) \notin \varphi([b])$. That is, $\varphi([a])\neq\varphi([b])$.
The proof of the  case 
$\{ i \in I \mid \varphi_i(a(i))\nsupseteq \varphi_i(b(i))\}\in \mathcal{F}$
is similar. 
\end{proof}

The function $\varphi$ from Lemma~\ref{lem:varphi-inj} is now shown to be 
a homomorphism. 

\begin{theorem}\label{thm:Pu(RDInFL)=RDInFL}
$\mathbb{P}_{\mathbb{U}}\left(\mathsf{RDInFL}\right) = \mathsf{RDInFL}$. 
\end{theorem}

\begin{proof}
We have to show that the ultraproduct $\mathbf{A}/\theta_{\mathcal{F}}$ of $\left\{\mathbf{A}_i \mid i \in I\right\}$ is representable. 
Recall that each of the $\mathbf{A}_i$ is embedded via $\varphi_i$ into an algebra of the form  $\mathfrak{D}\left( \langle E_i, \preceq_i \rangle \right)$. Each $\langle E_i, \preceq_i\rangle$ comes from a poset  $\mathbf{X}_i$ equipped with the necessary $E_i$ and $\alpha_i$. The definitions   after Theorem~\ref{thm:ultra_prod_DInFl} describe  the structure  $\langle E_Y, \preceq_Y\rangle $,
and, by Theorem~\ref{thm:ultraproductDInFL}(iv), $\left\langle\mathsf{Up}\left(\left\langle E_Y, \preceq_Y\right\rangle\right), \cap, \cup, \circ, \leqslant_Y,  \sim, -\right\rangle$ is a distributive InFL-algebra. 

From Lemma~\ref{lem:varphi-inj} we have that $\varphi$ is injective. 
We will now show that $\varphi$ is a homomorphism from $\mathbf{A}/\theta_{\mathcal{F}}$ into $\mathfrak{D}\left(\left\langle E_Y, \preceq_Y\right\rangle\right)$. We first show that $\varphi$ preserves meets:
\begin{eqnarray*}
&& \left([x], [y]\right) \in \varphi\left([a] \wedge^{\mathbf{A}/\theta_{\mathcal{F}}} [b]\right) \\
& \textnormal{ iff } & \left([x], [y]\right) \in \varphi\left(\left[a \wedge^{\mathbf{A}} b\right]\right)\\
& \textnormal{ iff } & \left\{i \in I \mid \left(x(i), y(i)\right) \in \varphi_i\left(\left(a \wedge^\mathbf{A} b\right)(i)\right) \right\} \in \mathcal{F}\\
& \textnormal{ iff } & 
\left\{i \in I \mid \left(x(i), y(i)\right) \in \varphi_i\left(a(i) \wedge^{\mathbf{A}_i} b(i)\right) \right\} \in \mathcal{F}\\
& \textnormal{ iff } & \left\{i \in I \mid \left(x(i), y(i)\right) \in \varphi_i\left(a(i)\right) \cap \varphi_i\left(b(i)\right) \right\} \in \mathcal{F}\\
%& \textnormal{ iff } & \left\{i \in I \mid \left(x(i), y(i)\right) \in \varphi_i\left(a(i)\right) \textnormal{ and } \left(x(i), y(i)\right) \in \varphi_i\left(b(i)\right) \right\} \in \mathcal{F}\\
& \textnormal{ iff } & \left\{i \in I \mid \left(x(i), y(i)\right) \in \varphi_i\left(a(i)\right)\right\} \cap \left\{i \in I \mid \left(x(i), y(i)\right) \in \varphi_i\left(b(i)\right) \right\} \in \mathcal{F}\\
& \textnormal{ iff } & \left\{i \in I \mid \left(x(i), y(i)\right) \in \varphi_i\left(a(i)\right)\right\} \in \mathcal{F} \textnormal{ and } \left\{i \in I \mid \left(x(i), y(i)\right) \in \varphi_i\left(b(i)\right) \right\} \in \mathcal{F}\\
& \textnormal{ iff } & \left([x], [y]\right) \in \varphi\left([a]\right) \textnormal{ and } \left([x], [y]\right) \in \varphi\left([b]\right) \\
& \textnormal{ iff } & \left([x], [y]\right) \in \varphi\left([a]\right) \cap \varphi\left([b]\right).
\end{eqnarray*}
The fourth equivalence above follows from the fact that $\varphi_i$ preserves meets for all $i\in I$. The left-to-right implication of the the third-to-last equivalence follows from the fact that $\mathcal{F}$ is closed under taking supersets, and the right-to-left implication follows from the fact that $\mathcal{F}$ is closed under taking intersections. 

Using the facts that $\mathcal{F}$ is an ultrafilter and $\varphi_i$ preserves joins for all $i \in I$, we can also show that $\varphi$ preserves joins using
a proof dual to the case for meets. 

To show that $\varphi$ preserves the monoid operation, first let $\left([x], [y]\right) \in \varphi\left([a]\right)\circ \varphi\left([b]\right)$. Then there is some $[z]\in Y$ such that $\left([x], [z]\right) \in \varphi\left([a]\right)$ and $\left([z], [y]\right) \in \varphi\left([b]\right)$. Hence, we obtain $\left\{i \in I \mid \left(x(i), z(i)\right) \in \varphi_i\left(a(i)\right)\right\} \in \mathcal{F}$ and $\left\{i \in I \mid \left(z(i), y(i)\right) \in \varphi_i\left(b(i)\right)\right\} \in \mathcal{F}$, which means 
\[
\left\{i \in I \mid \left(x(i), z(i)\right) \in \varphi_i\left(a(i)\right)\right\} \cap \left\{i \in I \mid \left(z(i), y(i)\right) \in \varphi_i\left(b(i)\right)\right\} \in \mathcal{F}.
\]
Since $\left\{i \in I \mid \left(x(i), z(i)\right) \in \varphi_i\left(a(i)\right)\right\} \cap \left\{i \in I \mid \left(z(i), y(i)\right) \in \varphi_i\left(b(i)\right)\right\}$ is contained in the set \break $\left\{i \in I \mid \left(x(i), y(i)\right) \in \varphi_i\left(a(i)\right) \circ \varphi_i\left(b(i)\right)\right\}$ and $\mathcal{F}$ is closed under taking supersets, it follows that $\left\{i \in I \mid \left(x(i), y(i)\right) \in \varphi_i\left(a(i)\right) \circ \varphi_i\left(b(i)\right)\right\} \in \mathcal{F}$. Now $\varphi_i\left(a(i)\right) \circ \varphi_i\left(b(i)\right) = \varphi_i\left(a(i) \cdot^{\mathbf{A}_i} b(i)\right) = \varphi_i\left(\left(a \cdot^\mathbf{A} b\right)(i)\right)$, so we have 
$\left\{i \in I \mid \left(x(i), y(i)\right) \in \varphi_i\left(\left(a \cdot^\mathbf{A} b\right)(i)\right)\right\} \in \mathcal{F}$. This shows that \break $\left([x], [y]\right) \in \varphi\left(\left[a\cdot^\mathbf{A} b\right]\right) = \varphi\left([a] \cdot^{\mathbf{A}/\theta_\mathcal{F}} [b]\right)$. 

For the reverse inclusion, let $\left([x], [y]\right) \in \varphi\left([a] \cdot^{\mathbf{A}/\theta_\mathcal{F}} [b]\right) = \varphi\left(\left[a\cdot^\mathbf{A} b\right]\right)$. Then we have \break $\left\{i \in I \mid \left(x(i), y(i)\right) \in \varphi_i\left(\left(a \cdot^\mathbf{A} b\right)(i)\right)\right\} \in \mathcal{F},$ and so 
$$\left\{i \in I \mid \left(x(i), y(i)\right) \in \varphi_i\left(a(i)\right)\circ \varphi_i\left(b(i)\right)\right\} \in \mathcal{F}.$$ 
Let $V = \left\{i \in I \mid \left(x(i), y(i)\right) \in \varphi_i\left(a(i)\right) \circ \varphi_i\left(b(i)\right)\right\}$. Since $\mathcal{F}$ is a proper filter, $V \neq \varnothing$. For each $i \in V$, there exists some $s_i \in X_i$ such that $\left(x(i), s_i\right) \in \varphi_i\left(a(i)\right)$ and $\left(s_i, y(i)\right) \in \varphi_i\left(b(i)\right)$. Now define an element of $\prod_{i\in I}X_i$ as follows:
$$z(i)=
\begin{cases} s_i & i \in V  \\ x(i)  & i \notin V
\end{cases}
$$
It then follows that $
V \subseteq \left\{i \in I \mid \left(x(i), z(i)\right) \in \varphi_i\left(a(i)\right)\right\} \cap \left\{i \in I \mid \left(z(i), y(i)\right) \in \varphi_i\left(b(i)\right)\right\}
$
Hence, since $\mathcal{F}$ is closed under taking supersets, we have 
$$
\left\{i \in I \mid \left(x(i), z(i)\right) \in \varphi_i\left(a(i)\right)\right\} \cap \left\{i \in I \mid \left(z(i), y(i)\right) \in \varphi_i\left(b(i)\right)\right\} \in \mathcal{F}
$$
which means 
$\left\{i \in I \mid \left(x(i), z(i)\right) \in \varphi_i\left(a(i)\right)\right\} \in \mathcal{F}$ and $\left\{i \in I \mid \left(z(i), y(i)\right) \in \varphi_i\left(b(i)\right)\right\} \in \mathcal{F}$. 
This shows $\left([x], [z]\right) \in \varphi\left([a]\right)$ and $\left([z], [y]\right) \in \varphi\left([b]\right)$, and therefore $\left([x], [y]\right) \in \varphi\left([a]\right) \circ \varphi\left([b]\right)$. 

Next, we show that $\varphi$ preserves the monoid identity:
\begin{eqnarray*}
\left([x], [y]\right) \in \varphi\left(1^{\mathbf{A}/\theta_{\mathcal{F}}}\right) & \textnormal{ iff } & \left([x], [y]\right) \in \varphi\left(\left[1^\mathbf{A}\right]\right)\\
& \textnormal{ iff } & \left\{i \in I \mid \left(x(i), y(i)\right) \in \varphi_i\left(1^\mathbf{A}(i)\right) \right\} \in \mathcal{F}\\
& \textnormal{ iff } & \left\{i \in I \mid \left(x(i), y(i)\right) \in \varphi_i\left(1^{\mathbf{A}_i}\right) \right\} \in \mathcal{F}\\
& \textnormal{ iff } & \left\{i \in I \mid x(i) \leqslant_i y(i) \right\} \in \mathcal{F}\\
& \textnormal{ iff } & [x] \leqslant_Y [y].
\end{eqnarray*}
Using Lemma~\ref{lem:alphaY_inverse}, we can also show that $\varphi$ preserves $\sim$:
\begin{eqnarray*}
 &&\left([x], [y]\right) \in \varphi\left({\sim}^{\mathbf{A}/\theta_\mathcal{F}}[a]\right) = \varphi\left(\left[{\sim}^{\mathbf{A}}a\right]\right)\\
& \textnormal{ iff } & \left\{i \in I \mid \left(x(i), y(i)\right) \in \varphi_i\left(\left({\sim}^\mathbf{A}a\right)(i)\right)\right\} \in \mathcal{F} \\
& \textnormal{ iff } & \left\{i \in I \mid \left(x(i), y(i)\right) \in \varphi_i\left({\sim}^{\mathbf{A}_i}a(i)\right)\right\} \in \mathcal{F} \\
& \textnormal{ iff } & \left\{i \in I \mid \left(x(i), y(i)\right) \in {\sim}\varphi_i\left(a(i)\right)\right\} \in \mathcal{F} \\
& \textnormal{ iff } & \left\{i \in I \mid \left(x(i), y(i)\right) \in \left(\varphi_i\left(a(i)\right)\right)^{c\smile}\circ \alpha_i\right\} \in \mathcal{F} \\
& \textnormal{ iff } & \left\{i \in I \mid \left(x(i), \alpha^{-1}_{i}\left(y(i)\right)\right) \in \left(\varphi_i\left(a(i)\right)\right)^{c\smile}\right\} \in \mathcal{F} \\
& \textnormal{ iff } & \left\{i \in I \mid \left(\alpha^{-1}_{i}\left(y(i)\right), x(i)\right) \in \left(\varphi_i\left(a(i)\right)\right)^{c}\right\} \in \mathcal{F}. 
\end{eqnarray*}
Again recalling that the complement 
$\left(\varphi_i\left(a(i)\right)\right)^{c}$ 
is taken in $E_i$, we get: 
\begin{eqnarray*}
&& \left\{i \in I \mid \left(\alpha^{-1}_{i}\left(y(i)\right), x(i)\right) \in \left(\varphi_i\left(a(i)\right)\right)^{c}\right\} \in \mathcal{F} \\
& \textnormal{ iff } & \left\{i \in I \mid \left(\alpha^{-1}_{i}\left(y(i)\right), x(i)\right) \in E_i \textnormal{ and } \left(\alpha^{-1}_{i}\left(y(i)\right), x(i)\right) \notin \varphi_i\left(a(i)\right)\right\} \in \mathcal{F} \\
& \textnormal{ iff } & \left\{i \in I \mid \left(\alpha^{-1}_{i}\left(y(i)\right), x(i)\right) \in E_i\right\} \cap \left\{ i \in I \mid  \left(\alpha^{-1}_{i}\left(y(i)\right), x(i)\right) \notin \varphi_i\left(a(i)\right)\right\} \in \mathcal{F} \\
& \textnormal{ iff } & \left\{i \in I \mid \left(\alpha^{-1}_{i}\left(y(i)\right), x(i)\right) \in E_i\right\} \in \mathcal{F} \textnormal{ and }  \left\{ i \in I \mid  \left(\alpha^{-1}_{i}\left(y(i)\right), x(i)\right) \notin \varphi_i\left(a(i)\right)\right\} \in \mathcal{F} \\
& \textnormal{ iff } & \left\{i \in I \mid \left(\alpha^{-1}_{i}\left(y(i)\right), x(i)\right) \in E_i\right\} \in \mathcal{F} \textnormal{ and }  \left\{ i \in I \mid  \left(\alpha^{-1}_{i}\left(y(i)\right), x(i)\right) \in \varphi_i\left(a(i)\right)\right\}^c \in \mathcal{F} \\
& \textnormal{ iff } & \left\{i \in I \mid \left(\alpha^{-1}_{i}\left(y(i)\right), x(i)\right) \in E_i\right\} \in \mathcal{F} \textnormal{ and }  \left\{ i \in I \mid  \left(\alpha^{-1}_{i}\left(y(i)\right), x(i)\right) \in \varphi_i\left(a(i)\right)\right\} \notin \mathcal{F} \\
& \textnormal{ iff } & \left\{i \in I \mid \left(\left(\alpha^{-1}_X\left(y\right)\right)(i), x(i)\right) \in E_i\right\} \in \mathcal{F} \textnormal{ and }  \left\{ i \in I \mid  \left(\left( \alpha^{-1}_X\left(y\right)\right)(i), x(i)\right) \in \varphi_i\left(a(i)\right)\right\} \notin \mathcal{F} \\
& \textnormal{ iff } & \left( [\alpha^{-1}_X(y)], [x]\right) \in E_Y \textnormal{ and } \left([\alpha^{-1}_X(y)], [x]\right) \notin \varphi\left([a]\right).
\end{eqnarray*}
Using the fact that the complement $\left(\varphi([a])\right)^c$ is taken in $E_Y$, we have  
\begin{eqnarray*}
\left( [\alpha^{-1}_X(y)], [x]\right) \in E_Y \textnormal{ and } \left([\alpha^{-1}_X(y)], [x]\right) \notin \varphi\left([a]\right)
& \textnormal{ iff } & \left([x], [\alpha^{-1}_X(y)]\right) \in \left(\varphi\left([a]\right)\right)^{c\smile}\\
& \textnormal{ iff } & \left([x], \alpha^{-1}_Y\left([y]\right)\right) \in \left(\varphi\left([a]\right)\right)^{c\smile}\\
& \textnormal{ iff } & \left([x], [y]\right) \in \left(\varphi\left([a]\right)\right)^{c\smile} \circ \alpha_Y = {\sim}\varphi\left([a]\right).
\end{eqnarray*}
The proof that $\varphi$ preserves $-$ is similar, and uses the definition of $\alpha_Y$. 
\end{proof}

Theorem~\ref{thm:Pu(RDInFL)=RDInFL} finds an immediate application in this final subsection. By combining it with  Definition~\ref{def:RSM}, we get the following result:  
\begin{corollary}\label{cor:}
$\mathbb{P_U}(\mathsf{RSM})=\mathsf{RSM}$.    
\end{corollary}
\begin{proof}
Consider $\{\, \mathbf{A}_i \mid i \in I \,\} \subseteq \mathsf{RSM}$. Each of the $\mathbf{A}_i$ is isomorphic to a direct reduct $\left(\mathbf{B}_i\right)_r$ 
of an algebra $\mathbf{B}_i \in \mathsf{RDInFL}$. That is, for each $i \in I$, there is a map $\varphi_i : A_i \to B_i$ which is an isomorphism from $\mathbf{A}_i$ to $\left(\mathbf{B}_i\right)_r$.

Let $\mathcal{F}$ be an ultrafilter on $I$ and recall that $\mathbf{A}$ and $\mathbf{B}$ are the direct products of the $\mathbf{A}_i$ and $\mathbf{B}_i$, respectively. 
(We denote by $A$ and $B$ the underlying sets of $\mathbf{A}$ and $\mathbf{B}$ respectively.)
By Theorem~\ref{thm:Pu(RDInFL)=RDInFL}, 
we have that $\mathbf{B}/\theta_{\mathcal{F}} \in \mathsf{RDInFL}$. We will show that 
$\mathbf{A}/\theta_\mathcal{F}$ is isomorphic to $\left(\mathbf{B}/\theta_\mathcal{F}\right)_r$.  

Let $ \varphi: A \to B$ be defined by $\varphi(a)(i)=\varphi_i(a(i))$. Clearly $\varphi$ is a bijective map. Next, consider $\overline{\varphi}: A/\theta_\mathcal{F} \to 
B/\theta_\mathcal{F}$, defined 
by $\overline{\varphi}([a])=[\varphi(a)]$. 
We show that $\overline{\varphi}$ is bijective. For $[b] \in B/\theta_{\mathcal{F}}$, 
since $b \in B$, there exists $a \in A$ with $\varphi(a)=b$. So $\overline{\varphi}\left( [a]\right)=\left[ \varphi(a) \right]=[b]$. If $a,c \in A$ with $[a]\neq [c]$, then $(a,c) \notin \theta_\mathcal{F}$. Thus
$\{\, i \in I \mid a(i) =c(i)\,\} \notin \mathcal{F}$, so 
$\{\, i \in I \mid a(i) \neq c(i)\,\} \in \mathcal{F}$. Since the $\varphi_i$ are injective, if $a(i)\neq c(i)$, then $\varphi_i(a(i))\neq \varphi_i (c(i))$ and so $\varphi(a)(i) \neq \varphi(c)(i)$. Hence 
$\{\, i \in I \mid a(i) \neq c(i) \,\} \subseteq \{\,i \in I \mid \varphi(a)(i) \neq \varphi(c)(i)\,\}$ so the latter set is in $\mathcal{F}$. Now 
$\{\,i\in I \mid \varphi(a)(i) = \varphi(c)(i)\,\} \notin \mathcal{F}$, so $(\varphi(a),\varphi(c)) \notin \theta_\mathcal{F}$ and therefore $[\varphi(a)]\neq [\varphi(c)]$, i.e. 
$\overline{\varphi}\left( [a] \right) \neq \overline{\varphi}\left( [c] \right)$.

Next, we will show that  $\overline{\varphi}$ preserves the fusion operation. For $[a], [c] \in A/\theta_{\mathcal{F}}$, we have 
$$\overline{\varphi}\left( [a] \cdot^{\mathbf{A}/\theta_{\mathcal{F}}}[c]\right) = \overline{\varphi}\left( \left[ a \cdot^{\mathbf{A}} c \right] \right) = \left[ \varphi \left(a \cdot^{\mathbf{A}} c\right) \right]. $$ 
For each $i\in I$, 
$\varphi_i : \mathbf{A}_i \to (\mathbf{B}_i)_r$ is an isomorphism. 
Also,  
$\cdot^{(\mathbf{B}_i)_r}=\cdot^{\mathbf{B}_i}$.  Hence for any $i \in I$: 
$$ \varphi(a \cdot^{\mathbf{A}} c)(i)= \varphi_i\left((a \cdot^{\mathbf{A}} c)(i)\right)=\varphi_i \left(a(i) \cdot^{\mathbf{A}_i} c(i) \right) 
=
\varphi_i(a(i))  \cdot^{(\mathbf{B}_i)_r} \varphi_i( c(i))  
=
\varphi_i(a(i))  \cdot^{\mathbf{B}_i} \varphi_i(c(i)). 
$$
Also observe that 
$\varphi_i(a(i))  \cdot^{\mathbf{B}_i} \varphi_i(c(i))  = 
\left( \varphi(a)\cdot^{\mathbf{B}} \varphi(c)\right)(i)
$. Using this in the last step below, we can now see that 
$\overline{\varphi}\left( [a]\right)
\cdot^{(\mathbf{B}/\theta_\mathcal{F})_r}
\overline{\varphi}\left([c]\right)=
\overline{\varphi}\left( [a]\right)
\cdot^{\mathbf{B}/\theta_\mathcal{F}}
\overline{\varphi}\left([c]\right)
=
[\varphi(a)]
\cdot^{\mathbf{B}/\theta_\mathcal{F}}
[\varphi(c)]=
\left[ \varphi(a) \cdot^{\mathbf{B}} \varphi(c) \right]=\left[ \varphi(a \cdot^{\mathbf{A}} c) \right]$ 
and hence 
$\overline{\varphi}\left( [a] \cdot^{\mathbf{A}/\theta_{\mathcal{F}}}[c]\right)=
\overline{\varphi}\left( [a]\right)
\cdot^{(\mathbf{B}/\theta_\mathcal{F})_r}
\overline{\varphi}\left([c]\right)$. 
Using the same approach, it is straightforward to show that $\overline{\varphi}$ preserves the operations $\wedge^{\mathbf{A}/\theta_{\mathcal{F}}}$,
$\vee^{\mathbf{A}/\theta_{\mathcal{F}}}$, $1^{\mathbf{A}/\theta_{\mathcal{F}}}$
and 
${\sim}^{\mathbf{A}/\theta_{\mathcal{F}}}$.

Now consider the operation 
$\to^{\mathbf{A}/\theta_{\mathcal{F}}}$ 
on $\mathbf{A}/\theta_{\mathcal{F}}$. 
For $[a]$, $[c] \in A/\theta_{\mathcal{F}}$, we have 

$$
\overline{\varphi} \left( [a]\to^{\mathbf{A}/\theta_{\mathcal{F}}} [c] \right) 
=\overline{\varphi}\left( [a \to^{\mathbf{A}} c] \right) =
\left[\varphi(a\to^{\mathbf{A}}c)
\right].
$$

For $i\in I$ and $b,d\in B_i$, recall that $b \to^{(\mathbf{B}_i)_r} d = {\sim}^{\mathbf{B}_i} \left( -^{\mathbf{B}_i} d \cdot^{\mathbf{B}_i} b \right)$. 
The third equality below uses the fact that   $\varphi_i : \mathbf{A}_i \to (\mathbf{B}_i)_r$ is an isomorphism. For any $i \in I$, we get: 
\begin{align*}
\varphi(a \to^{\mathbf{A}} c)(i)= \varphi_i\left((a \to^{\mathbf{A}} c)(i)\right)&=\varphi_i \left(a(i) \to^{\mathbf{A}_i} c(i) \right)\\ 
&=
\varphi_i(a(i))  \to^{(\mathbf{B}_i)_r} \varphi_i( c(i))  \\
&={\sim}^{\mathbf{B}_i} \left( 
-^{\mathbf{B}_i} \varphi_i(c(i)) \cdot^{\mathbf{B}_i} \varphi_i(a(i)) \right) \\
&= {\sim}^{\mathbf{B}} \left( -^\mathbf{B} \varphi(c) \cdot^{\mathbf{B}} \varphi(a) \right)(i). 
\end{align*}
Finally, we get: 
\begin{align*}
\overline{\varphi} \left( [a]\to^{\mathbf{A}/\theta_{\mathcal{F}}} [c] \right) 
=\overline{\varphi}\left( [a \to^{\mathbf{A}} c] \right) 
= \left[\varphi(a\to^{\mathbf{A}}c)
\right] 
&= \left[ {\sim}^{\mathbf{B}} \left( -^\mathbf{B} \varphi(c) \cdot^{\mathbf{B}} \varphi(a) \right) \right]\\
&= {\sim}^{\mathbf{B}/\theta_{\mathcal{F}}}\left[  \left( -^\mathbf{B} \varphi(c) \cdot^{\mathbf{B}} \varphi(a) \right) \right]\\
&= 
{\sim}^{\mathbf{B}/\theta_{\mathcal{F}}}\left(  \big[ -^\mathbf{B} \varphi(c)\big]  \cdot^{\mathbf{B}/\theta_\mathcal{F}} \big[ \varphi(a) \big]   \right)\\
&= 
{\sim}^{\mathbf{B}/\theta_{\mathcal{F}}}\left(   -^{\mathbf{B}/\theta_{\mathcal{F}}} \left[\varphi(c)\right]  \cdot^{\mathbf{B}/\theta_\mathcal{F}} \left[ \varphi(a) \right]   \right)\\
&= 
\left[ \varphi(a)\right ] \to^{(\mathbf{B}/\theta_\mathcal{F})_r}
\left[ \varphi(c) \right]\\
&= \overline{\varphi}\left([a]\right) \to^{(\mathbf{B}/\theta_{\mathcal{F}})_r} \overline{\varphi}\left([c]\right).
\end{align*}
\end{proof}

\subsection{Representing all Sugihara monoids}\label{subsec:RSM=SM}

In addition to $\mathsf{RSM}$ being closed under ultraproducts, it follows from the definition of $\mathsf{RDInFL}$ that $\mathsf{RSM}$ is closed under $\mathbb{I}$, $\mathbb{S}$ and $\mathbb{P}$. 
Next, we recall the following well-known theorem from universal algebra 
(cf. \cite[Theorem V-2.14]{BS2012}). 

\begin{theorem}\label{thm:alg-in-ultra}
Every algebra can be embedded in an ultraproduct of its finitely generated subalgebras. 
\end{theorem}

The finitely generated subalgebras of $\mathbf{S}=
\langle \mathbb{Z},\wedge,\vee,\cdot,\to,1, {\sim}\rangle 
$ are the algebras $\mathbf{S}_{2n+1}$ for $n \in \omega$. For the infinite Sugihara monoid $\mathbf{S}^*=
\langle \mathbb{Z}{\setminus}\{0\},\wedge,\vee,\cdot,\to,1, {\sim}\rangle$, its 
finitely generated subalgebras are $\mathbf{S}_{2n+2}$ for $n \in \omega$. 
Lastly, recall that the quasivariety of Sugihara monoids is generated by 
$\{\mathbf{S},\mathbf{S}^*\}$ (cf.~\cite[Section~3]{OR2007}). 
We now have everything in place in order to prove our main result. 

\begin{theorem}\label{thm:allSMrep}
Every Sugihara monoid is representable. That is, $\mathsf{SM} \subseteq \mathsf{RSM}$.   
\end{theorem}
\begin{proof}
We need to show that $\mathbb{ISP}(\mathbf{S},\mathbf{S}^*)=\mathsf{SM}$ is contained in $\mathsf{RSM}$. From
Theorem~\ref{thm:main}
and Examples~\ref{ex:S2},~\ref{ex:S3} and~\ref{ex:Maddux2010}, we know that $\{\mathbf{S}_i \mid i \in \mathbb{Z}^+\} \subseteq \mathsf{RSM}$. (The trivial algebra $\mathbf{S}_1$ is of course representable as it is a subalgebra of $\mathbf{S}_3$.)

Let $I= \{\,2k+1 \mid k \in \omega\,\}$ and  
$\mathbf{A}=\prod \{\,\mathbf{S}_i \mid i \in I\,\}$.  By Theorem~\ref{thm:alg-in-ultra}, there exists an ultrafilter $\mathcal{F}$ on $I$ and an embedding $e_1 : \mathbf{S} \hookrightarrow \mathbf{A}/\theta_{\mathcal{F}}$.  
Since $\mathbf{A}/\theta_{\mathcal{F}} \in \mathbb{P_U}(\mathsf{RSM})$ we get 
$\mathbf{S} \in \mathbb{ISP_U}(\mathsf{RSM})\subseteq \mathsf{RSM}$.

Next, let $J=\{\,2n+2\mid n \in \omega\,\}$ and $\mathbf{B}=\prod \{\,\mathbf{S}_j \mid j \in J \,\}$.  By Theorem~\ref{thm:alg-in-ultra}, there is an ultrafilter $\mathcal{G}$ on $J$ and an embedding 
$e_2: \mathbf{S}^* \hookrightarrow \mathbf{B}/\theta_{\mathcal{G}}$. Again, since $\mathbf{B}/\theta_{\mathcal{G}} \in\mathbb{P_U}(\mathsf{RSM})$ we get $\mathbf{S}^* \in \mathbb{ISP_U}(\mathsf{RSM})\subseteq \mathsf{RSM}$. 

Since $\mathsf{SM}=\mathbb{ISP}(\mathbf{S},\mathbf{S}^*)$ and 
$\{\mathbf{S},\mathbf{S}^*\} \subseteq \mathsf{RSM}$ we get 
$\mathsf{SM}=\mathbb{ISP}(\mathbf{S},\mathbf{S}^*) \subseteq \mathbb{ISP}(\mathsf{RSM})=\mathsf{RSM}$.
\end{proof}

\subsection*{Acknowledgements}
%%%%%%%%%%%%%%%%%%%%%%%%%%%%%%%%%%%%%%%%%%%%%%%%%%%%%%%%%
The authors are extremely grateful to the three anonymous referees for their 
very insightful and knowledgeable feedback on the earlier versions of this paper. Their comments enabled us to see that additional results could be possible, and this led to the development
of  Section~\ref{sec:RSM}. 
The first author would like to thank  Chapman University for its hospitality during a research visit in September/October 2023. We would also like to thank Jamie Wannenburg for useful discussions on this topic.

%%%%%%%%%%%%%%%%%%%%%%%%%%%%%%%%%%%%

\end{document}